\documentclass{article}

\usepackage[english]{babel}
\usepackage{arxiv}

\usepackage[utf8]{inputenc} 
\usepackage[T1]{fontenc}    
\usepackage{hyperref}       
\usepackage{url}            
\usepackage{booktabs}       
\usepackage{amsfonts}       
\usepackage{nicefrac}       
\usepackage{microtype}      
\usepackage{lipsum}
\usepackage{amsmath}
\usepackage{amsthm}
\usepackage[title]{appendix}
\usepackage{mathtools}
\usepackage{filecontents}
\usepackage{graphicx}
\usepackage{multirow}
\usepackage{amssymb}
\newtheorem{theorem}{Theorem}
\newtheorem{definition}{Definition}
\usepackage{cleveref}
\usepackage{csquotes}
\usepackage{float}
\usepackage{enumitem}

\newtheorem{lemma}{Lemma}
\newtheorem{remark}{Remark}
\newtheorem{proposition}{Proposition}

\newcommand{\Sign}{\text{Sign}}
\newcommand\norm[1]{\left\lVert#1\right\rVert}

\title{Stability analysis of a novel delay differential equation model of HIV infection of CD4\textsuperscript{+} T-cells}

\usepackage{biblatex}
\author{
  Hoang Anh NGO \\
  LIX, École Polytechnique\\
  Institut Polytechnique de Paris\\
  91120 Palaiseau, France \\
  
    \And
 Hung Dang NGUYEN \\
 Vietnamese - German Faculty of Medicine (VGFM) \\
 Pham Ngoc Thach University of Medicine \\
 Ho Chi Minh City 710000, Vietnam› \\
 
    \And
 Mehmet Dik \thanks{Corresponding author, Email address: \href{mailto:dikm@beloit.edu}{dikm@beloit.edu} }\\
  Department of Mathematics and Computer Science \\
  Beloit College\\
  WI 53511, United States \\
}

\begin{document}
\maketitle

\begin{abstract}
In this paper, we investigate a novel 3-compartment model of HIV infection of CD4\textsuperscript{+} T-cells with a mass action term by including two versions: one baseline ODE model and one delay-differential equation (DDE) model with a constant discrete time delay. 

Similar to various endemic models, the dynamics within the ODE model is fully determined by the basic reproduction term $R_0$. If $R_0 < 1$, the disease-free (zero) equilibrium will be asymptotically stable and the disease gradually dies out. On the other hand, if $R_0 > 1$, there exists a positive equilibrium that is globally/orbitally asymptotically stable within the interior of a predefined region.

To present the incubation time of the virus, a constant delay term $\tau$ is added, forming a DDE model. In this model, this time delay (of the transmission between virus and healthy cells) can destabilize the system, arising periodic solutions through Hopf bifurcation.

Finally, numerical simulations are conducted to illustrate and verify the results.
\end{abstract}


\section{Inroduction}

In the field of epidemiology, although our knowledge of viral dynamics  and virus-specific immmune responses has not fully developed, numerous mathematical models have been developed an investigated to describe the immunological response to HIV infection (for example, \cite{anderson1989complexhivmodel, culshaw2000hivmodel, deboer1998comparision, kepler1993bcells, percus1993sizetcells, mclean1990hivmodel} and references therein). The models have been used to explain different phenomena within the host body, and by directly applying the models to real clinical data, they can also predicts estimates of many measures, including the death rate of productively infected cells, the rate of viral clearance or the viral production rate.  

These simple HIV models have played an essential role in providing a better understanding in the dynamics of this infectious diseases, while providing very important biological meanings for the (combined) drug therapies used against it. For more references and detailed meta mathematical analysis on these models in general, we can refer to survey papers written by Kirschner, 1996 \cite{kirschner1996hivdynamics} or Perelson and Nelson, 1999 \cite{perelson1999hivmodelinvivo}

The simplest HIV model, only considering the dynamics of the virus concentration, is

\begin{equation}
    \frac{dV}{dt} = P - cV
\end{equation}

where

\begin{itemize}
    \item $P$ is an unknown function representing the rate of production of the virus,
    \item $V$ is the virus concentration.
\end{itemize}

The dynamics of the population of target cells (CD4\textsuperscript{+} T-cells for HIV or hepatic cells for HBV and HCV) is still not fully understood. Nevertheless, a reasonable, simple model for this population of cells, which can be extended further in various models, is

\begin{equation} \label{simple-T}
    \frac{dT}{dt} = s - dT + aT \left( 1 - \frac{T}{T_{\max}} \right) 
\end{equation}

with 

\begin{itemize}
    \item $s$ representing the rate at which new T-cells are created from sources within the body, such as the thymus, or from the proliferation of existing T-cells,
    \item $d$ being the death rate per T-cells,
    \item $a$ is the maximum proliferation rate of target T-cells, when the proliferation is represented by a logistic function, and
    \item $T_{\max}$ is the population density of T-cells at which proliferation shuts off.
\end{itemize}

Human immunodeficiency virus, or HIV, is a virus belonging to the genus Lentivirus, part of the family Retroviridae \cite{taxonomy}. It has an outer envelope of lipid and viral proteins, which encloses its core. The virion core contains two positive-sense single-stranded RNA and the enzyme reverse transcriptase, an RNA-dependent DNA polymerase. 

HIV, like most viruses, cannot reproduce by itself. Therefore, they require a host cell and its materials to replicate. For HIV, it infects a variety of immune cells, including helper T cells, lymphocytes, monocytes, and dendritic cells by attaching to a specific receptor called the CD4 receptor contained in the cell membrane. Along with a chemokine coreceptor, the virus is granted entry into the cell. Inside the host cell, the viral RNA is transcribed into DNA by the enzyme reverse transcriptase. However, the enzyme has no proofreading capacity, so errors often occur during this process, giving rise to 1 to 3 mutations per newly synthesized virus particle. The DNA provirus is then transported into the nucleus and inserts itself into the host cell DNA with the aid of viral integrase. Thus, the viral genetic code becomes a stable part of the cell genome, which is then transcribed into a full-length mRNA by the host cell RNA polymerase. The full-length mRNA would be 

\begin{enumerate}[label=(\alph*)]
    \item the genomes of progeny virus, which would be transported to the cytoplasm for assembly,
    \item translated to produce the viral proteins, including reverse transcriptase and integrase, and 
    \item spliced, creating new translatable sequences
\end{enumerate}

The nonstructural genes on the virus also encode regulatory proteins that have diverse effects on the host cell, including down-regulating host cell receptors like CD4 and major histocompatibility complex class I molecules, aiding in synthesizing full-length HIV RNAs and enabling transportation of the viral mRNAs out of the nucleus without being spliced by the host cell. Altogether, these effects enable viral mRNAs to be correctly translated into polypeptides and packaged into virions. These components are then transported to the plasma membrane and assembled into the mature virion, exiting the cell.

A person can contract the virus through one of four routes: sexual contact, either homo- or heterosexual; transfusions with whole blood, plasma, clotting factors and cellular fractions of blood; contaminated needles; perinatal transmission. The virus causes tissue destruction, immunodeficiency and can progress to acquired immunodeficiency syndrome (AIDS), completely breaking down the human body’s defense mechanisms. These patients are now more susceptible to infections that should be harmless to a normal person, such as P.jiroveci pneumonia or tuberculosis, and the conditions are worse as well. So far, treatments for the disease mainly target reverse transcriptase, viral proteases, and viral integration and fusion, dealing with the virus infection before it progresses to AIDS. Currently,  one treatment for HIV is highly active antiretroviral therapy (HAART), which includes a combination of drugs including nucleoside/nucleotide analog reverse transcriptase inhibitors, nonnucleoside reverse transcriptase inhibitors, protease inhibitors, fusion inhibitors, integrase inhibitors, and coreceptor blockers. These drugs are administered based on individualized criteria such as tolerability, drug-drug interactions, convenience/adherence, and possible baseline resistance. Although HAART can lower the viral load, the virus reemerges if the treatment is stopped. Therefore, HIV infection is currently both chronic and incurable. \cite{harvey2012microbiology}

Whenever the population reaches $T_{\max}$, it will decrease, allowing us to impose an upper constrain $d T_{\max} < s$. With this constrain, the equation \eqref{simple-T} has a unique equilibrium at

\begin{equation}
    \hat{T} = \frac{T_{\max}}{2a} \left[ a - d + \sqrt{(a - d)^2 + \frac{4as}{T_{\max}}}\right]
\end{equation}

In 1989, Perelson \cite{perelson1989hivmodelimmune} proposed a general model for the interaction between the human immune system and HIV; in the same paper, he also simplified that general model into a simpler model with four compartments, whose dynamics are described by a system of four ODEs:

\begin{itemize}
    \item Concentration of cells that are uninfected ($T$),
    \item Concentration of cells that are latently infected ($T^*$),
    \item Concentration of cells that are actively infected ($T^{**}$), and 
    \item Concentration of free infectious virus particles ($v$).
\end{itemize}

Later, he extended his own model in Perelson et al. (1993) \cite{perelson1993hivmodel} by proving various mathematical properties of the model, choosing parameter values from a restricted set that give rise to the long incubation period characteristic of HIV infection, and presenting some numerical solutions. He also observed that his model exhibits many clinical symptoms of AIDS, including:

\begin{itemize}
    \item Long latency period,
    \item Low levels of free virus in the environment, and
    \item Depletion of CD4\textsuperscript{+} cells.
\end{itemize}

Time delay, of one type or another, have been incoporated into biological models in various research papers (for example, \cite{perelson1989hivmodelimmune}); particularly, by the similar theoretical analysis to dynamical population system (in \cite{nowak1996populationdynamics}), they also play an important role in the dynamical properties of the HIV infection models. Generally speaking, systems of delay-differential equations (DDEs) have much more complicated dynamics than that of ordinary differential equations (ODEs), as the time delay can cause a stable equilibrium of the ODE system to become unstable, leading to the fluctuation of popolations. In studying the viral clearance rate, Perelson et al. (1996) \cite{perelson1996hivmodelinvivo} stated that there are two different types of delay that can occur in an HIV infection model:

\begin{itemize}
    \item \textbf{Pharmacological delay}: This delay occurs between the ingestion of drug and its appearance within cells,
    \item \textbf{Intracellular delay}: This delay happens between the initial HIV infection of a cell and the release of virions within the environment. 
\end{itemize}

There has also been various attempts by different authors, trying to come up with the most realistic model by implementing these delays, in one form or another (constant delay, discrete delay, continuous delay, etc.). For example,

\begin{itemize}
    \item Herz et al. (1996) \cite{herz1996viraldynamics}, who implemented a discrete delay to represent the intracellular delay in the HIV model. He showed that the incorporation of the delay would significantly shorten the estimate for the half-life of free virus particles. 
    \item Mittler et al. (1998) \cite{mittler1998hivmodeldelayed} stated that a $\gamma$ - distribution delay would be more realistic to describe the intracellular delay, then implemented it to the original model proposed by Perelson et al. (1996) \cite{perelson1996hivmodelinvivo}.
    \item Mittler et al. \cite{mittler1999improveddelay} et al. and Tam et al. (1999) \cite{tam1999delayeffect} also derived an analytic expression for the rate of decline of virus following drug treatment by assuming the drug to be completely effcacious.
    \item Song and Neumann (2007) \cite{song2007hivmodel} proposed a saturated mass-action term into the simplified model from Perelson et al. (1996) \cite{perelson1996hivmodelinvivo}, and later investigated the drug effectiveness under this saturation infection.
\end{itemize}

The paper will be organized as follows: First, we will investigate a simplified ODE model from Perelson et al. (1993) \cite{perelson1993hivmodel} by considering three main components: the uninfected CD4\textsuperscript{+} T-cells ($T$), the infected CD4\textsuperscript{+} T-cells ($I$), and the free virus ($V$) with. This model is also assumed to have a saturation response of the infection rate. Next, the existence and stability of he infected steady state are considered. Then, we incorporate a discrete constant delay into the model to indicate the time range between the infection of a CD4\textsuperscript{+} T-cell and the emission of viral particles at the cellular level, resulting in a system of three delay-differential equations (DDEs). To understand the dynamics of this delay model and obtain sufficient conditions for local/global asymptotic stability of the equilibria of all time delay, we carried out a complete analysis on the transcendental characteristic equations of the linearized system at both the viral-free equilibrium and the infected (positive) equilibrium. Finally, numerical simulations are carried out, using \texttt{Julia}, to confirm the obtained results, before some remarks are included in the conclusion. 

\section{The proposed ODE model}

Simplifying the model proposed in Perelson et al. (1993) \cite{perelson1993hivmodel} by reducing the number of dimensions and assuming that all of the infected cells have the ability of producing virus at an equal rate, we propose the following epidemic model of HIV infection of CD4\textsuperscript{+} T-cells as follows:

\begin{equation} \label{system}
    \begin{aligned}
    \frac{dT}{dt} & = s - dT + aT \left( 1 - \frac{T}{T_{\max}}\right) - \frac{\beta TV}{1 + \alpha V} + \rho I \\
    \frac{dI}{dt} & = \frac{\beta TV}{1 + \alpha V} - (\delta + \rho) I \\
    \frac{dV}{dt} & = qI - cV - k_1 VT
    \end{aligned}
\end{equation}

where 

\begin{itemize}
    \item $T(t)$ is the concentration of healthy CD4\textsuperscript{+} T-cells at time $t$ (target cells),
    \item $I(t)$ is the concentration of infected CD4\textsuperscript{+} T-cells at time $t$, and
    \item $V(t)$ is the viral load of the virions (concentration of free HIV at time $t$).
\end{itemize}

In infection modelling, it is very common to augment \eqref{system} with a "mass-action" term in which the rate of infection is given by $\beta TV$. This type of term is sensible, since the virus must interact with T-cells in order to infect and the probability of virus encountering a T-cell at a low concentration environment (where infected cells and viral load's motions are regarded as independent) can be assumed to be proportional to the product of the density, which is called linear infection rate. As a result, it follows that the classical models can assume that T-cells are infected at rate $-\beta TV$ and are generated at rate $\beta TV$.

With that simple mass-action infection term, the rates of change of uninfected cells, $T$, productively infected cells $I$, and free virus $V$, would be

\begin{equation} \label{system-linear}
    \begin{aligned} 
        \frac{dT}{dt} & = s - dT + aT \left( 1 - \frac{T}{T_{\max}} \right) - \beta TV \\
        \frac{dI}{dt} & = \beta T V - \delta I \\
        \frac{dV}{dt} & = qI - cV
    \end{aligned}
\end{equation}

Moreover, although the rate of infection in most HIV models is bilinear for the virus $V$ and the uninfected target cells $T$, the actual incidence rates are probably not strictly linear for each variable in over the whole valid range. For example, a non-linear or less-than-linear response in $V$ could occur due to the saturation at a high enough viral concentration, where the infectious fraction is significant for exposure to happen very likely. Thus, is it reasonable to assume that the infection rate of HIV modelling in saturated mass action is

\begin{equation}
\frac{\beta TV^x}{1 + \alpha V^y}, \quad x, y, \alpha > 0
\end{equation}

In this paper, we will investigate the viral model with saturation response of the infection rate where $x = y = 1$, for the sake of simplicity. With that being said, we will proceed to explain the parameters within the model, with 

\begin{itemize}
    \item $s$ is the rate at which new T-cells are created from source from precursors, 
    \item $d$ is the natural death rate of the CD4\textsuperscript{+} T-cells,
    \item $a$ is the maximum proliferation rate (growth rate) of T-cells (this means that $a > d$ in general),
    \item $T_{\max}$ is the T-cells population density at which proliferation shuts off (their carrying capacity),
    \item $\beta$ is the rate constant of infection of T-cells with free virus,
    \item $\rho$ is the "cure" rate, or the non-cytolytic loss of infected cells,
    \item $\delta$ is the death rate of the infected cells,
    \item $q$ is the reproduction rate of the infected cells, and 
    \item $c$ is the clearance rate constant (loss rate) of the virions.
\end{itemize}

From the explanations above, we can say that 

\begin{itemize}
    \item $\delta + \rho$ is the total rate of disappearance of infected cells from the environment,
    \item $1/\delta$ is the average lifespan of a productively infected cell
    \item $q/\delta$ is the total number of virions produced by an actively infected cell during its lifespan, and 
    \item $q$ is the average rate of virus released by each cell.
\end{itemize}

Under the absence of virus (i.e $I(t) = V(t) = 0 \quad \forall t > 0$), the T-cell population has a steady state value of

\begin{equation}
    \begin{aligned}
    T_0 & = \frac{T_{\max}}{2a} \left[ (a - d) + \sqrt{(a-d)^2 + \frac{4a}{T_{\max}}}\right]
    \end{aligned}
\end{equation}

The system \eqref{system} needs to be initialized with the following initial conditions

\begin{equation}
    T(0) > 0, \quad I(0) > 0, \quad V(0) > 0,
\end{equation}

which leads us to denote that

\begin{equation}
    R_+^3 = \{ (T,I,V) \in \mathbb{R}^3 \| T \geq 0, I \geq 0, V \geq 0 \}
\end{equation}

\subsection{Equilibria and local stability}

The system \eqref{system} has two steady states: the uninfected steady state $E_0 = \left( T_0, 0, 0 \right)$ and the (positive) infected steady state $\bar{E} = \left( \bar{T}, \bar{I}, \bar{V} \right)$, where:

\begin{equation}
    \begin{aligned}
        \bar{T} & = \frac{T_{\max}}{2a} \left[ a - d - \delta \frac{q \beta - (\delta + \rho)}{q\alpha (\delta + \rho)} + \sqrt{\left( a - d - \delta \frac{q \beta - (\delta + \rho)}{q\alpha (\delta + \rho)}\right)^2 - \frac{4a}{T_{\max}} \left( \frac{\delta c}{q \alpha} - s \right)}\right] \\
        \bar{I} & = \frac{[q\beta - (\delta + \rho)k_1] \bar{T} - (\delta + \rho)c}{q\alpha (\delta + \rho)}\\
        \bar{V} & = \frac{1}{\alpha} \left[ \frac{q\beta \bar{T}}{\alpha (\delta + \rho)(c_1 + k_1 T} - 1\right]\\
    \end{aligned}
\end{equation}

Now, we will proceed to analyse the stability of the equilibria of system \eqref{system}.

Since $T_0$ and $\bar{T}$ satisfy

\begin{equation}
    \begin{aligned}
        s - dT_0 + aT_0 \left( 1 - \frac{T_0}{T_{\max}}\right) & = 0 \\
        s - d\bar{T} + a\bar{T} \left( 1 - \frac{\bar{T}}{T_{\max}}\right) & = \delta \bar{I} = \frac{\delta}{q \alpha (\delta + \rho)} \left[ (q\beta - (\delta + \rho)) T - (\delta + \rho)c \right]
    \end{aligned}
\end{equation}

we get that

\begin{equation}
    \bar{T} > \frac{c(\delta + \rho)}{q \beta - (\delta + \rho)k_1} \quad \Rightarrow \quad s - d \bar{T} + a \bar{T} \left( 1 - \frac{\bar{T}}{T_{\max}} \right) > 0 \quad \Rightarrow \quad T_0 > \bar{T}
\end{equation}

and 

\begin{equation}
    \bar{T} < \frac{c(\delta + \rho)}{q \beta - (\delta + \rho)k_1} \quad \Rightarrow \quad s - d \bar{T} + a \bar{T} \left( 1 - \frac{\bar{T}}{T_{\max}} \right) < 0 \quad \Rightarrow \quad T_0 < \bar{T}
\end{equation}

Hence,

\begin{itemize}
    \item If $\bar{T} > \frac{c(\delta + \rho)}{q\beta - (\delta + \rho)k_1}$, then $T_0 > \bar{T} > \frac{c(\delta + \rho)}{q\beta - (\delta + \rho)k_1}$, which means that $E_0(T_0, 0, 0)$ is unstable, while the positive equilibrium $\bar{E}(\bar{T}, \bar{I}, \bar{V})$ exists. 
    \item If $\bar{T} < \frac{c(\delta + \rho)}{q\beta - (\delta + \rho)k_1}$, then $T_0 < \bar{T} < \frac{c(\delta + \rho)}{q\beta - (\delta + \rho)k_1}$, which means that $E_0(T_0, 0, 0)$ is locally asymptotically stable, while the positive equilibrium $\bar{E}(\bar{T}, \bar{I}, \bar{V})$ is not feasible, as $\bar{I} < 0, \bar{V} < 0$.
\end{itemize}

Let 
\begin{equation}
    R_0 =  \left( \frac{q\beta - (\delta + \rho) k_1 }{c(\delta + \rho)} \right) \bar{T}
\end{equation}

We can see that $R_0$ is the bifurcation parameter. When $R_0 < 1$, the uninfected steady state $E_0$ is stable and the infected steady state $\bar{E}$ does not exist (unphysical). When $R_0 > 1$, $E_0$ becomes unstable and $\bar{E}$ exists.

For system \eqref{system-linear}, it is known that the basic reproductive ratio is given by:

\begin{equation}
    R_{01} = \left( \frac{q\beta - (\delta + \rho) k_1 }{c(\delta + \rho)} \right) T_0
\end{equation}

Once again, we emphasize the large difference of the basic reproduction ratio between the linear infection rate and the saturation infection rate. 
\begin{itemize}
    \item If $\alpha \to 0$, then $\bar{T} \to \frac{c(\delta + \rho)}{q\beta - (\delta + \rho)}, \quad R_0 \to 1$;
    \item If $\alpha \rightarrow +\infty$, then $\bar{T} \to T_0, R_0 \to R_{01}$.
\end{itemize}

The \texttt{Jacobian matrix} of system \eqref{system} is:

\begin{equation}
    \left( \begin{matrix}
    (a - d) - \frac{2aT}{T_{\max}} - \frac{\beta V}{1 + \alpha V} & \rho & - \frac{\beta T}{(1 + \alpha V)^2} \\
    \frac{\beta V}{1 + \alpha V} & - (\delta + \rho) & \frac{\beta T}{(1 + \alpha V)^2} \\
    -k_1 V & q & - c - k_1 T
    \end{matrix} \right)
\end{equation}

Let $E^*(T^*, I^*, V^*)$ be any arbitrary equilibrium. Then, the characteristic equation about $E^*$ is:

\begin{equation} \label{jac}
    \left| \begin{matrix} 
    \lambda + \left( (d - a) + \frac{2aT^*}{T_{\max}} + \frac{\beta V^*}{1 + \alpha V^*}\right) & - \rho & \frac{\beta T^*}{(1 + \alpha V^*)^2} \\
    - \frac{\beta V^*}{1 + \alpha V^*} & \lambda + (\delta + \rho) & - \frac{\beta T^*}{(1 + \alpha V^*)^2} \\
    k_1 V^* & - q & \lambda + (c + k_1 T^*)
    \end{matrix} \right| = 0
\end{equation}

For equilibrium $E_0 = (T_0, 0, 0)$, \eqref{jac} reduces to

\begin{equation}
    \left( \lambda - a + d + \frac{2aT_0}{T_{\max}} \right) \left[ \lambda ^2 + (c + \delta + \rho) \lambda + c(\delta + \rho) - q\beta T_0 \right] = 0
\end{equation}

Hence, we can see that $E_0 = (T_0, 0, 0)$ is locally asymptotically stable if $R_0 < 1$, and it is a saddle point if $\dim W^s(E_0) = 2$, or if $\dim W^s(E_0) = 1$ while $R_0 > 1$. As a result, we have the following theorems

\begin{theorem} 
If $R_0 < 1$, $E_0 = (T_0, 0, 0)$ is locally asymptotically stable; else, if $R_0 > 1$, $E_0 = E_0 = (T_0, 0, 0)$ is unstable.
\end{theorem}

\begin{theorem} \label{boundedness}
There exists $M > 0, M \in \mathbb{R}$ such that for any positive solution $(T(t), I(t), V(t))$ of system \eqref{system}, 
\begin{equation}
    T(t) \leq M, I(t) \leq M, V(t) \leq M
\end{equation}
for all large enough $t$.
\end{theorem}

\begin{proof} 

Let $L(t) = T(t) + I(t)$ and assume that $L(0) = T(0) + I(0) = \text{const} = c$. Calculating the derivative of $L(t)$ using the equations in system \eqref{system}, we have:

\begin{equation} \label{ineq-lt}
    \begin{aligned}
        \frac{dL(t)}{dt} & = \frac{dT(t)}{dt} + \frac{dI(t)}{dt} \\
        & = s - dT + aT \left( 1 - \frac{T}{T_{\max}}\right) - \delta I \\
        & = -dt - \delta I - \frac{a}{T_{\max}} \left( T - \frac{T_{\max}}{2a} \right)^2 + \frac{4s + aT_{\max}}{4} \\
         & \leq - (T+I) \min{(d, \delta)} - \frac{a}{T_{\max}} \left( T - \frac{T_{\max}}{2a} \right)^2 + \frac{4s + aT_{\max}}{4} \\
         & = - hL(t) - M_0 \left( h = \min{(d, \delta)}, M_0 = \frac{4s + aT_{\max}}{4} \right)
    \end{aligned}
\end{equation}

Let $U(t) = L(t) - \frac{M_0}{h}$. This means that \begin{equation}
    \begin{aligned}
        U(0) & = L(0) - \frac{M_0}{h} = c - \frac{M_0}{h} \\
        \frac{dU(t)}{dt} & = \frac{dL(t)}{dt}
    \end{aligned}
\end{equation}

The inequality \eqref{ineq-lt} can be rewritten as

\begin{equation}
    \frac{dU(t)}{dt} \leq (-h) U(t)
\end{equation}

which yields, according to Gronwall's inequality,

\begin{equation}
    \begin{aligned}
        U(t) & \leq U(0) \exp \left( \int_0^t (-h) ds \right) \\
        & = \left( c - \frac{M_0}{h} \right) \exp \left( \left[ -hs \right]_0^t \right) \\
        & = \left( c - \frac{M_0}{h} \right) \exp(-ht) \\
        & \leq c - \frac{M_0}{h}
    \end{aligned}
\end{equation}

or 

\begin{equation}
    T(t) + I(t) = L(t) = U(t) + \frac{M_0}{h} = c - \frac{M_0}{h} + \frac{M_0}{h} = c
\end{equation}

As $T(t) > 0, I(t) > 0 \ \forall i \in \mathbb{Z}^+$, we can say that
\begin{equation}
    V(t) \leq c, I(t) \leq c  
\end{equation}

Moreover, we also know that

\begin{equation} \label{ineq-vt}
    \frac{dV}{dt} = qI - cV - k_1 VT \leq qI - cV \leq qc - cV = -c(V - q)
\end{equation}

Setting $V(0) = \text{const} = c_V$, using the exact same procedure with Gronwall's inequality, we obtain 

\begin{equation}
    V(t) \leq c_V \ \forall t \in \mathbb{Z}^+
\end{equation}

With $M = \max{(c, c_V)}$, we would then conclude that

\begin{equation}
    T(t) \leq M, I(t) \leq M, V(t) \leq M \ \forall t \in \mathbb{Z}^+
\end{equation}

We can easily see that this set is convex. As a consequence, the system \eqref{system} is dissipative.

The proof is complete.
\end{proof}

From this theorem, we define

\begin{equation} \label{defiD}
    D = \left\{ (T, I, V) \in \mathbb{R}^3, 0 \leq T, I, V \leq M \right\}.
\end{equation}

Denote

\begin{equation}
    M = d - a + \frac{2a\bar{T}}{T_{\max}}, \quad N = \frac{\beta \bar{V}}{1 + \alpha \bar{V}}, \quad P = \frac{\beta \bar{T}}{(1 + \alpha \bar{V})^2}.
\end{equation}

Then, the characteristic equation of the system around the equilibrium $\bar{E}(\bar{T}, \bar{I}, \bar{V})$ reduces to:

\begin{equation} \label{characeq}
    \lambda^3 + a_1 \lambda^2 + (a_2 + a_4) \lambda + (a_3 + a_5) = 0
\end{equation}

where

\begin{equation} \label{define-as}
    \begin{aligned} 
    a_1 & = M + (\delta + \rho + c_1 + k_1 \bar{T}) \\
    a_2 & = (\delta + \rho) (c_1 + k_1 T) + M(\delta + \rho + c_1 + k_1 \bar{T}) + (-k_1 \bar{V}P) \\
    a_3 & = \rho \left[ -N(c_1 + k_1 \bar{T}) + Pk_1 \bar{V} \right] + PNq \\
    a_4 & = -NP \\
    a_5 & = M(\delta + \rho)(c_1 + k_1 \bar{T}) - P(\delta + \rho) k_1 \bar{V}
    \end{aligned}
\end{equation}

By the Routh-Hurwitz criterion \cite{kuang1993dde}, it follows that all eigenvalues of equation \eqref{characeq} have negative real parts if and only if

\begin{equation}
    a_1 > 0, \quad a_3 + a_5 > 0, \quad a_1(a_2 + a_4) - (a_3 + a_5) > 0
\end{equation}

This leads us to the following theorem

\begin{theorem} \label{asymptotically-stable}
Suppose that
\begin{enumerate}
    \item $R_0 > 1$,
    \item $a_1 > 0, \quad a_3 + a_5 > 0, \quad a_1(a_2 + a_4) - (a_3 + a_5) > 0$.
\end{enumerate}
Then, the positive equilibrium $\bar{E}(\bar{T}, \bar{I}, \bar{V})$ is asymptotically stable.
\end{theorem}

\begin{theorem} \label{globallyasympstable<1}
If $R_0 < 1$, then $E_0(T_0, 0, 0)$ is globally asymptotically stable.
\end{theorem}

\begin{proof}
First of all, as $R_0 < 1$, we would have

\begin{equation}
    T_0 < \bar{T} < \frac{c(\delta + \rho)}{q\beta - (\delta + \rho)}
\end{equation}

which means that 

\begin{equation}
    p < \frac{(c+ k_1 T)(\delta + \rho)}{\beta T}
\end{equation}
From the system \eqref{system}, we would have

\begin{equation}
    \begin{aligned}
    \frac{dI}{dt} & \leq \beta TV - (\delta + \rho) I, \\
    \frac{dV}{dt} & = qI - cV - k_1 VT.
    \end{aligned}
\end{equation}

Now, we would consider the following comparative system

\begin{equation} \label{system-compa-z}
    \begin{aligned}
    \frac{dz_1}{dt} & = \beta T z_2 - (\delta + \rho) z_1 \\
    \frac{dz_2}{dt} & = pz_1 - cz_2 - k_1 z_2 T
    \end{aligned}
\end{equation}

We will consider the following form of Lyapunov function:

\begin{equation}
    L(\mathbf{X}) = V(z_1,z_2) = \frac{\delta + \rho}{(\beta T)^2} z_1^2 + \frac{1}{c + k_1 T}z_2^2
\end{equation}

The derivative of the function can be calculated as follows

\begin{equation}
    \begin{aligned}
    \frac{dL}{dt} & = \frac{\partial L}{\partial z_1} \frac{dz_1}{dt} + \frac{\partial L}{\partial z_2} \frac{dz_2}{dt} \\
    & = 2 \frac{\delta + \rho}{(\beta T)^2} z_1 \left( \beta T z_2 - (\delta + \rho) z_1 \right) + 2 \frac{1}{c + k_1 T} z_2 \left( qz_1 - c z_2 - k_1 Tz_2 \right) \\
    & = -2 \left[ \left( \frac{\delta + \rho}{\beta T} z_1\right)^2 + z_2^2 - \left ( \frac{\delta + \rho}{\beta T} z_1 z_2 + \frac{q}{c + k_1 T} \right) z_1 z_2 \right] \\
    & \leq -2 \left[ \left( \frac{\delta + \rho}{\beta T} z_1 \right)^2 + z_2^2 - \left ( \frac{\delta + \rho}{\beta T} + \frac{\beta + \rho}{\beta T} \right) z_1 z_2\right] \\
    & = - 2 \left[ \frac{\delta + \rho }{\beta T} z_1 - z_2 \right]^2 \leq 0 \quad \forall z_1, z_2 
    \end{aligned}
\end{equation}

We can see that the derivative is negative definite everywhere except at $(0,0)$. This means that $(z_1, z_2) = (0,0)$ is globally asymptotically stable.

As we can also see that 

\begin{equation}
    0 \leq I(0) \leq z_1(0), \quad 0 \leq V(0) \leq z_2(0)
\end{equation}

which means that, if the system \eqref{system-compa-z} admits the initial values $(z_1(0), z_2(0))$, we have that

\begin{equation}
    I(t) \leq z_1(t), \quad V(t) \leq z_2(t) \quad \forall t > t_1
\end{equation}

or, in other words,

\begin{equation} \label{lim-iv}
    \lim_{t \to + \infty} I(t) = \lim_{t \to + \infty} V(t) = 0 
\end{equation}

From this, using the first equation of the system \eqref{system}, for an $\epsilon \  in (0,1)$ infinitesimal,

\begin{equation}
    s + (a - d - \delta \epsilon) T - \frac{aT^2}{T_{\max}} \leq \frac{dT(t)}{dt} \leq s + (a - d) T - \frac{aT^2}{T_{\max}} \quad \forall t > t_2
\end{equation}

which shows that

\begin{equation} \label{lim-t}
\lim_{t \to + \infty} T(t) = T_0. 
\end{equation}

From \eqref{lim-iv} and \eqref{lim-t}, we conclude that the system is globally asymptotically stable. The proof is complete.
\end{proof}

\begin{theorem} \label{permanent}
If $R_0 > 1$, then the system \eqref{system} is permanent.
\end{theorem}

\begin{proof}
If $R_0 > 1$, we would have 
\begin{equation}
    (q\beta - (\delta + \rho) k_1) T_0 > (q\beta - (\delta + \rho) k_1) \bar{T} > c(\delta + \rho) 
\end{equation}

We will proceed to prove the weak permanence of this system using contradiction.

Assume that the system is not weakly permanent, from Theorem \ref{globallyasympstable<1}, there exists a positive orbit $(T(t), I(t), V(t))$ such that

\begin{equation} \label{orbit-origin}
    \lim_{t \to + \infty} T(t) = T_0, \quad \lim_{t \to + \infty} I(t) = \lim_{t \to + \infty} V(t) = 0
\end{equation}

Since $T_0 > \frac{c(\delta + \rho)}{q \beta - (\delta + \rho)}$, combining with \eqref{orbit-origin}, we choose an arbitrary infinitesimal $\epsilon > 0$ such that there exists a $t_0 > 0$, for all $t > t_0$,

\begin{equation}
    \begin{aligned}
    \frac{T_0 - \epsilon}{1 + \alpha \epsilon} & > \frac{c(\delta + \rho)}{q\beta - (\delta + \rho)} \\
    T(t) & > T_0 - \epsilon,  \\
    V(t) & < \epsilon  
    \end{aligned}
\end{equation}

Under these conditions, the system \eqref{system} becomes 
\begin{equation}
    \begin{aligned}
    \frac{dI}{dt} & = \frac{\beta TV}{1 + \alpha V} - (\delta + \rho) I \geq \frac{\beta (T_0 - \epsilon)}{1 + \alpha \epsilon} V - (\delta + \rho) I(t) \\
    \frac{dV}{dt} & = qI - (c_1 + k_1 T) \approx qI - cV - k_1 T_0
    \end{aligned}
\end{equation}

Consider the following Jacobian matrix

\begin{equation}
    J_\epsilon = \left( \begin{matrix}
    - (\delta + \rho) & \frac{\beta (T_0 - \epsilon)}{1 + \alpha \epsilon} \\
    q & -(c + k_1 T_0)
    \end{matrix} \right)
\end{equation}

Since $J_\epsilon$ has positive off-diagonal element, according to the Perron - Frobenius theorem, for the maximum positive eigenvalue $j_1$ of $J_\epsilon$, there is an associated positive eigenvector $v = \left( \begin{matrix} v_1 \\ v_2 \end{matrix} \right)$.

Next, we consider a system associated with the Jacobian matrix $J_\epsilon$

\begin{equation} \label{system-z-eps}
    \begin{aligned}
    \frac{dz_1}{dt} & = \frac{\beta (T_0 - \epsilon)}{1 + \alpha \epsilon} z_2 - (\delta + \rho) z_1 \\
    \frac{dz_2}{dt} & = qz_1 - (c + k_1 T_0) z_2
    \end{aligned}
\end{equation}

Let $z(t) = (z_1(t), z_2(t))$  be a solution of \eqref{system-z-eps} through $(lv_1, lv_2)$ at $t = t_0$, where $l > 0$ satisfies that

\begin{equation}
    lv_1 < I(t_0), \quad lv_2 < V(t_0)
\end{equation}

As we know that the semi-flow of \eqref{system-z-eps} is monotone and $J_\epsilon v = v > 0$, $z_i(t) (t = 1,2)$ is strictly increasing, meaning $\lim_{t \ to +\infty} z_i(t) = + \infty$. This contradicts the Theorem \eqref{boundedness}, saying that the positive solution of \eqref{system} is bounded from above. This contradiction says that there exists no positive orbit of \eqref{system} tends to $(T_0, 0, 0)$ and $t \to +\infty$. Combining this and a result provided in \cite{butler1986uniformpersistence}, we conclude that the system \eqref{system} is permanent.

The proof is complete.

\end{proof}

\begin{theorem} \label{globally-asymptotically-stable-intD}
Assume that $D$ is convex and bounded. Suppose that the system

\begin{equation}
    \frac{dX}{dt} = F(X), \quad X \in D
\end{equation}

is competitive, permanent and has the property of stability of periodic orbits. If $\bar{X_0}$ is the only equilibrium point in $\text{int}D$ and if it is locally asymptotically stable, then it is globally asymptotically stable in $\text{int}{D}$.
\end{theorem}

\begin{proof}
This matrix can easily be proven by considering the Jacobian matrix and choose the matrix $H$ as

\begin{equation}
H = \left( \begin{matrix}
1 & 0 & 0 \\
0 & -1 & 0 \\
0 & 0 & 1
\end{matrix} \right)
\end{equation}

By simple calculation, we obtain that

\begin{equation}
    H \frac{\partial f}{\partial x} H =
    \left( \begin{matrix}
    (a - d) - \frac{2aT}{T_{\max}} - \frac{\beta V}{1 + \alpha V} & - \rho & - \frac{\beta T}{(1 + \alpha V)^2} \\
    - \frac{\beta V}{1 + \alpha V} & - (\delta + \rho) & - \frac{\beta T}{(1 + \alpha V)^2} \\
    -k_1 V & - q & - c - k_1 T
    \end{matrix} \right)
\end{equation}

This means that the system \eqref{system} is competitive in $D$, with respect to the partial order defined by the orthant

\begin{equation}
    K = \left\{ (T, I, V) \in \mathbb{R}^3 \| T \leq 0, I \geq 0, V \geq 0 \right\}
\end{equation}
\end{proof}

\begin{remark}
As $D$ is convex and the system \eqref{system} is competitive in $D$. we can say that the system \eqref{system} satisfies the \textbf{Poincare - Bendixson property}. This has been proven by Hirsch (1990) \cite{hirsch1990systems}, Zhu and Smith (1994) \cite{zhu1994stableperiodicorbits} and Smith and Thieme (1991) \cite{smith1991semiflows} that any three-dimensional competitive system that lie in convex sets would have the Poincaré - Bendixson property; in other words, any non-empty compact omega limit set that contains no equilibria must be a closed orbit. 
\end{remark}

\begin{theorem} \label{globally-asymptotically-stable}
Let $c = I(0) + T(0)$ and suppose that
\begin{enumerate}
    \item $R_0 > 1$,
    \item $a_1 > 0, a_3 + a_5 > 0, a_1 (a_2 + a_4) - (a_3 + a_5) > 0$.
\end{enumerate}
Then, the positive equilibrium $\bar{E}(\bar{T}, \bar{I}, \bar{V})$ of system \eqref{system} is globally asymptotically stable provided that one of the following two assumptions hold

\begin{enumerate}
    \item $T_{\max} \frac{a - d + k_1 c}{2a} < m < T_0 < T_{\max} \frac{a - d + \delta + k_1 c }{2a}$,
    \item $m > T_{\max}\frac{a - d + \delta + k_1 c}{2a}$.
\end{enumerate}
\end{theorem}

As we have already known that the system \eqref{system} is competitive and permanent (from Theorem \ref{permanent} and Theorem \ref{globally-asymptotically-stable-intD}), while $\bar{E}(\bar{T}, \bar{I}, \bar{V})$ is locally asymptotically stable if the two properties (i) and (ii) of Theorem \ref{globally-asymptotically-stable} holds. As a result, in accordance with Theorem \ref{globally-asymptotically-stable-intD} (choosing $D = \Omega$), Theorem \ref{globally-asymptotically-stable} if we can prove that the system \eqref{system} has the stability of periodic orbits. We will proceed to prove this under the following proposition.

\begin{proposition} \label{stability-periodic-orbits}
Assume condition (iii) or (iv) of \ref{globally-asymptotically-stable} hold true. Then, system \eqref{system} has the property of stability of periodic orbits.
\end{proposition}

\begin{proof}
Let $P(t) = ((T(t), I(t), V(t))$ be a periodic solution whose orbit $\Gamma$ is contained in $\text{int} \Omega$. In accordance with the criterion given by Muldowney in \cite{muldowney1990compoundmat}, for the asymptotic orbital stability of a periodic orbit of a general autonomous system, it is sufficient to prove that the linear non-autonomous system
\begin{equation} \label{system-linear-nonauto}
    \frac{dW(t)}{dt} = \left( DF^{[2]}\left( P(t) \right) \right) W(t)
\end{equation}

is asymptotically stable, where $DF^{[2]}$ is the second additive compound matrix of the Jacobian $DF$ \cite{song2007hivmodel}. 

The Jacobian matrix of the system \eqref{system} is given by

\begin{equation}
    J = \left( \begin{matrix}
    (a - d) - \frac{2aT}{T_{\max}} - \frac{\beta V}{1 + \alpha V} & \rho & - \frac{\beta T}{(1 + \alpha V)^2} \\
    \frac{\beta V}{1 + \alpha V} & - (\delta + \rho) & \frac{\beta T}{(1 + \alpha V)^2} \\
    - k_1 V & q & - (c + k_1 T)
    \end{matrix} \right)
\end{equation}

For the solution $P(t)$, the equation \eqref{system-linear-nonauto} becomes

\begin{equation}
    \begin{aligned} \label{system-w}
    \frac{dW_1}{dt} & = - \left( \delta + \rho - (a - d) + \frac{2aT}{T_{\max}} + \frac{\beta V}{1 + \alpha V} \right) W_1 + \frac{\beta T}{(1 + \alpha V)^2} (W_2 + W_3), \\
    \frac{dW_2}{dt} & = qW_1 + \left( a - d - \frac{2aT}{T_{\max}} - \frac{\beta V}{1 + \alpha V} - (c + k_1 T) \right) W_2 + \rho W_3, \\
    \frac{dW_3}{dt} & = k_1 V W_1 + \frac{\beta V}{1 + \alpha V} W_2 - (\delta + \rho + c + k_1 T) W_3.
    \end{aligned}
\end{equation}

To prove that the system \eqref{system-w} is asymptotically stable, we shall use the following Lyapunov function, which is similar to the one found in \cite{li1995globalstabilityseir} for the SEIR model:

\begin{equation}
    L(W_1(t), W_2(t), W_3(t), T(t), I(t), V(t)) = \norm{\left( 
    W_1(t), \frac{I(t)}{V(t)} W_2(t), \frac{I(t)}{V(t)} W_3(t) \right)}
\end{equation}

where $\norm{ \cdot }$ is the norm in $\mathbb{R}^3$ defined by

\begin{equation}
    \norm{(W_1, W_2, W_3)} = \sup \{ |W_1|, |W_2 + W_3| \}
\end{equation}

From Theorem \ref{permanent}, we obtain that the orbit of $P(t)$ remains at a positive distance from the boundary of $\Omega$. Therefore,

\begin{equation}
    I(t) \geq \eta, \quad V(t) \geq \eta, \quad \eta = \min\{\underbar{I}, \underbar{V}\} \quad \forall t \to +\infty
\end{equation}

Hence, the function $L(t)$ is well defined along $P(t)$ and 

\begin{equation} \label{ineq_L}
    L(W_1, W_2, W_3; T, I, V) \geq \frac{\eta}{M} \norm{(W_1, W_2, W_3)}
\end{equation}

Along a solution $(W_1, W_2, W_3)$ of the system \eqref{system-w}, $L(t)$ becomes

\begin{equation}
    L(t) = \sup \left\{ \left| W_1(t)\right|, \frac{I(t)}{V(t)} \left( \left| W_2(t) \right| + \left| W_3(t) \right| \right) \right\}
\end{equation}

Then, we would have the following inequalities

\begin{equation}
    \begin{aligned}
        D_+ |W_1(t)| & \leq - \left( \delta + \rho - (a - d) + \frac{2aT}{T_{\max}} + \frac{\beta V}{1 + \alpha V} \right) |W_1| + \frac{\beta T}{(1 + \alpha V)^2} \left( |W_2(t)| + |W_3(t)| \right) \\
        D_+ |W_2(t)| & \leq q |W_1(t)| + \left( a - d - \frac{2aT}{T_{\max}} - \frac{\beta V}{1 + \alpha V} - (c + k_1 T)\right) |W_2(t)| + \rho |W_3(t)| \\
        D_+|W_3(t)| & \leq k_1 V |W_1(t)| + \frac{\beta V}{1 + \alpha V} |W_2(t)| - (\delta + \rho + c + k_1 T) |W_3(t)|
    \end{aligned}
\end{equation}

From this, we get

\begin{equation}
    \begin{aligned}
    D_+ \frac{I}{V} (|W_2| + |W_3|) & = \left( \frac{dI/dt}{V} - \frac{I dV/dt}{V^2}\right) (|W_2| + |W_3|) + \frac{I}{V} D_+ (|W_2| + |W_3|) \\
    & \leq \left( \frac{dI/dt}{I} - \frac{dV/dt}{V}\right) \frac{I}{V} (|W_2| + |W_3|) + \left( \frac{qI}{V} + k_1 I \right) |W_1| \\
    & - \left(- a + d + \frac{2aT}{T_{\max}} + (c+k_1 T) \right) \frac{I}{V} |W_2(t)| - (\delta + c + k_1 T) \frac{I}{V} |W_3(t)|
    \end{aligned}
\end{equation}

Thus, we can obtain

\begin{equation}
    D_+ L(t) \leq \sup \{ g_1(t), g_2(t) \} L(t),
\end{equation}

where

\begin{equation}
    \begin{aligned}
        g_1(t) & = - \delta - \rho + a - d - \frac{2aT}{T_{\max}} - \frac{\beta V}{1 + \alpha V} + \frac{\beta TV}{I(1 + \alpha V)^2}\\
        g_2(t) & = \frac{qI}{V} + k_1 I + \frac{dI/dt}{I} - \frac{dV/dt}{V} - G_1 \\
        G_1 & = \min \left\{ -a + d + \frac{2aT}{T_{\max}} + (c + k_1 T), \delta + c + k_1 T \right\}
    \end{aligned}
\end{equation}

From the second equation of the system \eqref{system}, we obtain

\begin{equation}
    \begin{aligned}
        g_1(t) & = - \delta - \rho + a - d - \frac{2aT}{T_{\max}} - \frac{\beta V}{1 + \alpha V} + \frac{\beta T V}{I (1 + \alpha V)^2} \\
        & \leq - \delta - \rho + a - d - \frac{2aT}{T_{\max}} - \frac{\beta V}{1 + \alpha V} + \frac{\beta t V}{I(1 + \alpha V)} \\
        & = a - d - \frac{2aT}{T_{\max}} - \frac{\beta T}{1 + \alpha V} + \frac{dI/dt}{I}
    \end{aligned}
\end{equation}

Here, we consider two different cases.

\begin{itemize}
    \item \textbf{Case 1:} If Point 3 of Theorem \ref{globally-asymptotically-stable} holds, then
        \begin{equation}
            - \delta < a - d - \frac{2aT}{T_{\max}} < 0,
        \end{equation}
        
        that is

        \begin{equation}
            G_1 = -a + d + \frac{2aT}{T_max} + (c + k_1 T)
        \end{equation}
        
        Then, we would obtain
        \begin{equation}
            g_2(t) = a - d - \frac{2aT}{T_{\max}} + k_1 I + \frac{dI/dt}{I} = g_1(t) + k_1 I + \frac{\beta V}{1 + \alpha V} > g_1(t)
        \end{equation}
        
        Hence, 
        \begin{equation} \label{ineq-mu1}
            \sup \{ g_1(t), g_2(t) \} \leq a - d - \frac{2aT}{T_{\max}} + k_1 I + \frac{dI/dt}{I} \leq - \mu_1 + \frac{dI/dt}{I}
        \end{equation}
        
        where
        
        \begin{equation}
            \mu_1 > 0, \quad a - d - \frac{2aT}{T_{\max}} + k_1 I \leq - \mu_1 < 0
        \end{equation}
        
        with the assumption that $k_1 I$ is negligible compare to the term $a - \frac{2aT}{T_{\max}}$. This assumption would be verified in the examples of the simulation part below.
        
        \item \textbf{Case 2:} If Point 4 of Theorem \ref{globally-asymptotically-stable} holds, then
        \begin{equation}
            - a + d + \frac{2aT}{T_{\max}} \leq \delta,
        \end{equation}
        
        which means that $G_1 = \delta + c + k_1 T$. Then, we obtain that
        
        \begin{equation} \label{ineq-mu2}
            \mu_2 < 0, \quad g_1(t) < g_2(t) = k_1 T - \delta + \frac{dI/dt}{I} \leq - \mu_2 + \frac{dI/dt}{I}
        \end{equation}
        
        with the same assumption that $k_1T < \sigma$ in reasonably practical scenarios. Hence, 
        
        \begin{equation}
            \sup \{ g_1(t), g_2(t) \} \leq - \mu + \frac{dI/dt}{I}
        \end{equation}
        
        Let $\mu = \min \{ \mu_1, \mu_2 \}$. Then, form \eqref{ineq-mu1} and \eqref{ineq-mu2}, we have
        
        \begin{equation}
            \sup \{ g_1(t), g_2(t) \} \leq - \mu +  \frac{dI/dt}{I},
        \end{equation}
        
        or 
        
        \begin{equation}
            D_+ L(t) \leq \left( - \mu + \frac{dI/dt}{I} \right) L(t).
        \end{equation}
        
        According to Gronwall's inequality, we would have
        
        \begin{equation}
            \begin{aligned}
            L(t) & \leq L(0) \exp \left( \int_0^t \left[ - \mu + \frac{dI/dt}{I} \right] ds \right) \\
            & = L(0) \exp \left( \left[ -\mu s + \ln(I(s)) \right]_0^t \right) \\
            & = L(0) \exp(-\mu t) \exp \left( \ln(I(t)) - \ln(I(0)) \right) \\
            & = L(0) \exp(-\mu t) \frac{I(t)}{I(0)} \\
            & \leq \frac{M L(0)}{I(0)} \exp(-\mu t) \to 0 \quad \text{as} \quad t \to +\infty
            \end{aligned}
        \end{equation}
\end{itemize}

From \eqref{ineq_L}, we conclude that 

\begin{equation}
    \left( W_1(t), W_2(t), W_3(t) \right) \to 0 \quad \text{as} \quad t \to + \infty 
\end{equation}

This implies that the linear system equation \eqref{system-w} is asymptotically stable, and therefore the periodic solution is asymptotically orbitally stable. This proves proposition \ref{stability-periodic-orbits}

\end{proof}

\begin{theorem} \label{orbitally-asymptotically-stable}
Suppose that
\begin{enumerate}
    \item $R_0 > 1$,
    \item $a_1 > 0, a_3 + a_5 > 0, a_1(a_2 + a_4) - (a_3 + a_5) > 0$.
\end{enumerate}
Then, system \eqref{system} has an orbitally asymptotically stable periodic solution.
\end{theorem}

\begin{proof}
First, we perform a change of variables as follows:

\begin{equation}
    z_1(t) = -T(t), \quad z_2(t) = I(t), \quad z_3(t) = -V(t)
\end{equation}

Applying this transformation to the system \eqref{system}, we obtain

\begin{equation} \label{system-z}
    \begin{aligned}
        \frac{dz_1(t)}{dt} & = -s - dz_1 + az_1 \left( 1 + \frac{z_1}{T_{\max}}\right) + \frac{\beta z_1 z_3}{1 - \alpha z_3} + \rho z_2 \\
        \frac{dz_2(t)}{dt} & = \frac{\beta z_1 z_3}{1 - \alpha z_3} - (\delta + \rho) z_2 \\
        \frac{dz_3(t)}{dt} & = - q z_2 -c z_3 + k_1 z_1 z_3
    \end{aligned}
\end{equation}

The Jacobian matrix of the system \eqref{system-z} is then given by

\begin{equation}
    J(z) = \left( \begin{matrix}
    a - d + \frac{2az_1}{T_{\max}} + \frac{\beta z_3}{1 - \alpha z_3} & \rho & \frac{\beta z_1}{(1 + \alpha z_3)^2} \\
    \frac{\beta z_3}{1 - \alpha z_3} & - (\delta + \rho) & \frac{\beta z_1}{(1 + \alpha z_3)^2} \\
    k_1 z_3 & -q & -c + k_1 z_1 
    \end{matrix} \right)
\end{equation}

Similar to the definition of the set $D$ at \ref{defiD}, we define set $E$ as:

\begin{equation}
    E = \left\{ (z_1, z_2, z_3): z_1 \leq 0, z_2 \geq 0, z_3 \leq 0 \right\}
\end{equation}

Since $J(z)$ has non-positive off diagonal elements at each point of $E$, \eqref{system-z} is competitive at $E$. Set $z^* = (-T^*, I^*, V^*)$. It is easy to see that $z^*$ is unstable and $\det J(z^*) < 0$. Furthermore, it follows from Theorem \ref{permanent} that there exists a compact set $B$ in the interior of $E$ such that for any $z_0 \in \text{int} E$, there exists $T(z_0) > 0$ such that $z(t, z_0) \in B$ for all $t > T(z_0)$. Consequently, by Theorem 1.2 in Zhu and Smith (1994) \cite{zhu1994stableperiodicorbits} for the class of three-dimensional competitive systems, it has an orbitally asymptotically stable periodic solution.

The proof is complete.
\end{proof}

\section{The delay differential equation (DDE) model}
In this section, we introduce a time delay into the model \eqref{system-linear} to represent the incubation time that the vectors need to become infectious. The model is rewritten as follows

\begin{equation} \label{system-delay}
    \begin{aligned}
        \frac{dT}{dt} & = s - dT + aT \left( 1 - \frac{T}{T_{\max}} \right) - \frac{\beta TV}{1 + \alpha V} + \rho I \\
        \frac{dI}{dt} & = \frac{\beta T(t - \tau) V(t - \tau)}{1 + \alpha V(t - \tau)} - (\delta + \rho) I \\
        \frac{dV}{dt} & = qI - cV - k_1 VT
    \end{aligned}
\end{equation}

under the initial values 

\begin{equation}
    T(\theta) = T_0, \quad I(\theta) = I_0, \quad V(\theta) = V_0 \quad \forall \theta \in [-\tau, 0]
\end{equation}

All parameters of this delay model are the same as those of the system \eqref{system}, except that the additional positive constant $\tau$ represents the length of the delay, in days.

This time delay parameter can be explained as follows: At time $t$, only healthy cells that have been infected by the virus $\tau$ days ago (i.e at time $t - \tau$ are infectious, provided that they have survived the incubation period of $\tau$ days and were alive at the time $t - \tau$ when they infect the healthy cells. As a result, the incidence term of healthy cells in the derivative of infected cells with respect to time is modified from $\beta T(t)V(t)$ to $\beta T(t - \tau) V(t-\tau)$.

The reproduction of this delay differential equation can be given the same as the original ODE model, which is

\begin{equation}
    R_{01} = \frac{q\beta - (\delta + \rho)k_1}{c(\delta + \rho)} T_0
\end{equation}

Its biological meaning is that, if one virus is introduced in the population of uninfected cells, the total number of secondary infected cells during the infectious period would be 
$\frac{q\beta - (\delta + \rho)k_1}{c(\delta + \rho)}$.


\subsection{Local and Global Stability of the Disease-free Equilibrium}

Within this section, we will study the local and global stability of the disease-free equilibrium $E_0$ of the delay model in two cases: when $R_0 > 1$ and when $R_0 < 1$.

\begin{theorem}
The disease-free equilibrium of the system \eqref{system-delay} is locally asymptotically stable if $R_0 < 1$, and is unstable if $R_0 > 1$.
\end{theorem}

\begin{proof}
Linearizing the system \eqref{system-delay} around $E_0 = (T_0, 0, 0)$, we obtain one negative characteristic root

\begin{equation}
    \lambda_1 = a - d - \frac{2aT_0}{T_{\max}}
\end{equation}

and the following transcendental characteristic equation whose roots are the remaining eigenvalues

\begin{equation} \label{characeq-delay-diseasefree}
    \lambda^2 + (\delta + \rho + c + k_1 T_0) + (c + k_1 T_0)(\delta + \rho) - q \beta T_0 e^{-\lambda \tau} = 0
\end{equation}

For $\tau = 0$, we obtain the exact same quadratic equation as the original ODE system. In this case, we have proven previously that all eigenvalues of the characteristic equation \eqref{characeq-delay-diseasefree} have negative real parts. According to the Routh - Hurwitz criterion, the disease free equilibirum $E_0$ will be locally asymptotically stable when $R_0 < 1$ and is unstable when $R_0 > 1$. 

As a result, we now only need to prove that the statement holds true for all $\tau \neq 0$.

\begin{itemize}
    \item \textbf{Case 1:} $R_0 > 1$. In this case, we expect that \eqref{characeq-delay-diseasefree} has one positive root and the disease-free equilibrium is unstable. Indeed, we arrange the characteristic equation into the form of
    
    \begin{equation}
        \lambda^2 + (\delta + \rho + c + k_1 T_0) \lambda = q\beta T_0 e^{-\lambda \tau} - (c + k_1 T_0)(\delta + \rho) 
    \end{equation}
    
    Now, suppose that $\delta \in \mathbb{R}$ and denote 
    
    \begin{equation}
        \begin{aligned}
        F(\lambda) & = \lambda^2 + (\delta + \rho + c + k_1 T_0) \lambda \\
        G(\lambda) & = q\beta T_0 e^{-\lambda \tau} - (c + k_1 T_0)(\delta + \rho)
        \end{aligned}
    \end{equation}
    
    We would then have that
    \begin{equation}
        F(0) = 0, \quad \lim_{\lambda \to +\infty} F(\lambda) = +\infty
    \end{equation}
    while
    \begin{equation}
        \begin{aligned}
            & G(0) = q \beta T_0 - (c + k_1 T_0) (\delta + \rho) = c (\delta + \rho) (R_0 - 1) > 0, \\
            & G'(\lambda) < 0 \quad \forall \lambda > 0 
        \end{aligned}
    \end{equation}
    
    As a result, the two functions must intersect at a point $\lambda > 0$, which means that the equation \eqref{characeq-delay-diseasefree} admits a positive real root, which means that the disease-free equilibrium is unstable.
    
    \item \textbf{Case 2:} $R_0 < 1$. First, we can notice that \eqref{characeq-delay-diseasefree} can not have any non-negative roots since , while 
    
    \begin{equation}
        \begin{aligned}
        F(0) & = 0, \quad F'(\lambda) > 0 \quad \forall \lambda \geq 0 \\
         G(\lambda) & < 0 , \quad G'(\lambda) < 0 \quad \forall \lambda > 0
        \end{aligned}
    \end{equation}
    
    As a result, if \eqref{characeq-delay-diseasefree} has roots with non-negative real parts, they must be complex and should be obtained from a pair of complex conjugate with cross the imaginary axis. This means that \eqref{characeq-delay-diseasefree} must have a pair of purely imaginary roots for $\tau > 0$. 
    
    As a result, we assume that $\lambda = i\omega$, and without loss of generality, we assume that $\omega > 0$ is a root of \eqref{characeq-delay-diseasefree}, meaning that
    
    \begin{equation}
        -\omega^2 + i\omega(\delta + \rho + c + k_1 T_0) + (c + k_1 T_0)(\delta + \rho) - q\beta T_0 (\cos(\omega \tau) + i \sin(\omega \tau)) = 0
    \end{equation}
    
    Separating the real and imaginary part, we would have
    
    \begin{equation}
        \begin{aligned}
            -\omega^2 + (c + k_1 T_0)(\delta + \rho) & = q \beta T_0 \cos(\omega \tau) \\
            (\delta + \rho + c + k_1 T_0) \omega & = -q \beta T_0 \sin(\omega \tau)
        \end{aligned}
    \end{equation}
    
    Squaring and adding up both sides of the two equations above, we obtain the following fourth-order equation for $\omega$ as
    
    \begin{equation} \label{eq-omega-diseasefree}
        \omega^4 + \omega^2 \left[ (\delta + \rho + c + k_1 T_0)^2 - 2(c + k_1 T_0)(\delta + \rho) \right] + [(c + k_1 T_0) (\delta + \rho)]^2 - (q \beta T_0)^2 = 0 
    \end{equation}
    
    To reduce this fourth-order equation into a quadratic equation, let $z = \omega^2$ and denote the coefficients as
    
    \begin{equation}
        \begin{aligned}
        a_1 & = (\delta + \rho + c + k_1 T_0)^2 - 2(c + k_1 T_0)(\delta + \rho) \\
        a_2 & = (c + k_1 T_0)^2 (\delta + \rho)^2 - (q \beta T_0)^2
        \end{aligned}
    \end{equation}
    
    The equation \eqref{eq-omega-diseasefree} can be rewritten as
    
    \begin{equation} \label{eq-z-diseasefree}
        z^2 + a_1 z + a_2 = 0
    \end{equation}
    
    Since $R_0 < 1$, we would have that
    
    \begin{equation}
        \begin{aligned}
            a_1 & = (\delta + \rho)^2 + (c + k_1 T_0)^2 > 0 \\
            a_2 & = [(c + k_1 T_0) (\delta + \rho) + q \beta T_0] c(\delta + \rho)(1 - R_0) > 0 
        \end{aligned}
    \end{equation}
    
    As a result, this means that the two roots of \eqref{eq-z-diseasefree} have positive product, which means that they have the same sign, regardless of being real or complex. As these two roots also have negative real products, they would be either negative real numbers, or complex conjugate with negative real parts. As a result, the equation \eqref{eq-z-diseasefree} can not have any positive real roots, leading to the fact that there would be no $\omega$ such that $i\omega$ is a root of \eqref{characeq-delay-diseasefree}. Using Rouche's theorem, we conclude that the real parts of all eigenvalues of the characteristic equation of the disease-free equilibrium \eqref{characeq-delay-diseasefree} are all negative for all delay values $\tau > 0$.
    
    In conclusion, if $R_0 < 1$, the disease-free equilibrium $E_0$ is locally asymptotically stable.
\end{itemize}

The proof is complete.
\end{proof}

\subsection{Local and Global Stability of the Positive Equilibrium}

To study the stability of the steady states $\bar{E}$, we define

\begin{equation}
    x(t) = T(t) - \bar{T}, \quad y(t) = I(t) - \bar{I}, \quad z(t) = V(t) - \bar{V}. 
\end{equation}

Then, the linearized system of \eqref{system-delay} at $\bar{E}$ is given by

\begin{equation} \label{system-delay-linear}
    \begin{aligned} 
        \frac{dx(t)}{dt} & = \left[ -d + a - \frac{2a\bar{T}}{T_{\max}} - \frac{\beta \bar{V}}{(1+\alpha \bar{V})^2}\right] x(t) + \rho y(t) - \frac{\beta \bar{T}}{(1 + \alpha \bar{V})^2} z(t) \\
        \frac{dy(t)}{dt} & = \frac{\beta \bar{V}}{1 + \alpha \bar{V}} x(t - \tau) - (\delta + \rho) y(t) + \frac{\beta \bar{V}}{(1 + \alpha \bar{V})^2} z(t - \tau) \\
        \frac{dz(t)}{dt} & = -k_1 \bar{V} x(t) + q y(t) - (c + k_1 \bar{T}) z(t)
    \end{aligned}
\end{equation}

The system \eqref{system-delay-linear} can be expressed in matrix form as follows

\begin{equation}
    \frac{d}{dt} \left( \begin{matrix} x(t) \\ y(t) \\ z(t) \end{matrix}\right) = A_1 \left( \begin{matrix} x(t) \\ y(t) \\ z(t) \end{matrix}\right) + A_2 \left( \begin{matrix} x(t - \tau) \\ y(t - \tau) \\ z(t - \tau) \end{matrix}\right),
\end{equation}

where $A_1$ and $A_2$ are $3 \times 3$ matrices given by

\begin{equation}
    A_1 = \left( \begin{matrix}
    a - d - \frac{2a\bar{T}}{T_{\max}} & \rho & - \frac{\beta \bar{T}}{(1 + \alpha \bar{V})^2} \\
    0 & -(\delta + \rho) & 0 \\
    -k_1 \bar{V} & q & -(c+k_1\bar{T})
    \end{matrix} \right), \quad 
    A_2 = \left( \begin{matrix}
    0 & 0 & 0 \\
    \frac{\beta \bar{V}}{1+ \alpha \bar{V}} & 0 & \frac{\beta \bar{T}}{(1 + \alpha \bar{V})^2} \\
    0 & 0 & 0
    \end{matrix}\right).
\end{equation}

The characteristic equation of system \eqref{system-delay-linear} is given by

\begin{equation}
    \Delta(\lambda) = \left| \lambda I - A_1 - e^{-\lambda \tau} A_2 \right| = 0, 
\end{equation}

that is,

\begin{equation} \label{characeq-delay}
    \lambda^3 + a_1 \lambda^2 + a_2 \lambda + a_5 = - e^{-\lambda \tau} (a_3 + a_4 \lambda),
\end{equation}

with $a_i (i = 1, ..., 5)$ previously defined in \eqref{define-as}.

Next, we shall study the distribution of the roots of the transcendental equation \eqref{characeq-delay} with respect to $0$ analytically. Based on the point (or assumption) that the positive steady state of the original ODE model \eqref{system} is stable, we will derive further conditions on the parameters to ensure that the steady state of the delay model is still stable.

First, we will consider the base case when $\tau = 0$. Then, the characeristic equation \eqref{characeq-delay} will become \eqref{characeq-delay-diseasefree}. Now, we will assume that all roots of his equation, in this case, has all negative real parts, which is equivalent to the fact that the conditions in Theorem \ref{orbitally-asymptotically-stable} are satisfied. As the delay term $\tau$ is considered to be continuous on $\mathbb{R}$, from Rouché's Theorem \cite{dieudonne1960modernanalysis}, the transcendental equaion \eqref{characeq-delay} can only have roots with negative real parts if and only if it has purely imaginary roots. We will investigate whether \eqref{characeq-delay} can admit any purely imaginary roots; from which, we will be able to determine the conditions under which all eigenvalues would have negative real parts.

We assume that $\lambda = \eta(\tau) + i \omega(\tau) \quad (\omega > 0)$ is the eigenvalue of the characteristic equation \eqref{characeq-delay}, where $\eta(\tau)$ and $\omega(\tau)$ are functions depending on the delay term $\tau$. As the positive equilibrium $\bar{E}$ of the model \eqref{system} is stable, we can say that $\eta(0) < 0$ at $\tau = 0$. 

If $\eta (\tau_0) = 0$ for some certain values of of $\tau_0 > 0$ (which means that $\lambda = i\omega(\tau_0)$ are purely imaginary roots of the characteristic equation \eqref{characeq-delay}, the steady state $\bar{E}$ would lose is stability and become unstable whenever $\eta(\tau_0)$ is greater han $0$. In other words, if there exists no $\omega(\tau_0)$ such that the condition above happens, or, if the characteristic equation \eqref{characeq-delay} does not have any purely imaginary roots for all values of $\tau$, the positive equilibrium $\bar{E}$ is always stable. In the following part, we will prove that this statement is indeed correct for equation \eqref{characeq-delay}.

Clearly $i\omega \ (\omega > 0)$ is a root of equation \eqref{characeq-delay} if and only if 

\begin{equation}
    -i \omega^3 - a_1 \omega^2 + i a_2 \omega + a_5 = - a_3 (\cos(\omega \tau) - i \sin(\omega \tau)) - a_4 \omega( \sin(\omega \tau) + i cos(\omega \tau))
\end{equation}

Separating the real and imaginary parts, we would have

\begin{equation} \label{system-realim}
    \begin{aligned}
    a_1 \omega^2 - a_5 & = a_3 \cos(\omega \tau) + a_4 \omega \sin(\omega \tau) \\
    \omega^3 - a_2 \omega & = -a_3 \sin(\omega \tau) + a_4 \omega \cos(\omega \tau)
    \end{aligned}
\end{equation}

Squaring both sides of each equation above and adding up, we obtain the following sixth-degree equation for $\omega$:

\begin{equation} \label{eq-omega}
    \omega^6 + (a_1^2 - 2a_2) \omega^4 + (a_2^2 - 2a_1a_5 - a_4^2) \omega^2 + (a_5^2 - a_3^2) = 0
\end{equation}

Since this equation contains only even powers of $\omega$, we can reduce the order by letting once again $z = \omega^2$ and 

\begin{equation}
    \begin{aligned}
        m_1 & = a_1^2 - 2a_2, \\
        m_2 & = a_2^2 - 2a_1a_5 - a_4^2, \\
        m_3 & = a_5^2 - a_3^2, 
    \end{aligned}
\end{equation}

the equation \eqref{eq-omega} becomes

\begin{equation} \label{eq-z}
    h(z) = z^3 + m_1z^2 + m_2z + m_3 = 0
\end{equation}

In order to show that the positive equilibrium $\bar{E}$ is locally stable, we have to prove that the equation \eqref{eq-z} does not have any positive real root which associates to the square of $\omega$; that is, \eqref{characeq-delay} can not have any purely imaginary roots. The Theorem below provides us with necessary conditions satisfying the result.

\begin{theorem} 
If $m_3 \geq 0$ and $m_2 > 0$, the equation \eqref{eq-omega} has no positive real roots.
\end{theorem}

\begin{proof}
We will proceed to prove the lemma above using contradiction.

Assume that there exists at least one positive real roots for the equation $h(z) = 0$. 

Notice that $h(0) = m_3 \geq 0$. This means that in order for the equation \eqref{eq-omega} to have a positive real roots, there exists $z_0 \geq 0$ such that
\begin{equation}
    \frac{dh(z_0)}{dz_0} \leq 0.    
\end{equation}

This is equivalent to 
\begin{equation}
    \frac{dh(z_0)}{dz_0} = 3z_0^2 + 2m_1z_0 + m_2 \leq 0
\end{equation}

or

\begin{equation}
    \frac{ - m_1 - \sqrt{m_1^2 - 3m_2}}{2} \leq z \leq \frac{-m_1 + \sqrt{m_1^2 - 3m_2}}{2} < 0
\end{equation}

This contradicts our original assumption that $z_0 \geq 0$, which means that there does not exist any $z_0 \geq 0$ such that $\frac{dh(z_0)}{dz_0} > 0$, or the equation does not have any positive real roots.

The proof is complete.
\end{proof}

This theorem has implied that there exists no $\omega$ such that $i\omega$ is an eigenvalue of the characteristic function \eqref{characeq-delay}. As a result, from Rouche's theorem \cite{dieudonne1960modernanalysis}, the real parts of the eigenvalues of \eqref{characeq-delay} are negative for all $\tau \geq 0$. Summarizing all the above analysis, we have the following theorem

\begin{theorem} \label{positive-equi-asymptotically-stable-delay}
Suppose that
\begin{enumerate}
    \item $a_1 > 0, a_3 + a_5 > 0, a_1(a_2+a_4) - (a_3+a_5) > 0$;
    \item $m_3 \geq 0$ and $m_2 > 0$.
\end{enumerate}
Then, the infected steady state $\bar{E}$ of the delay model \eqref{system-delay} is absolutely stable; that is, $\bar{E}$ is asymptotically stable for all $\tau \geq 0$.
\end{theorem}

\begin{remark}
The Theorem \ref{positive-equi-asymptotically-stable-delay} indicates that if the parameters satisfy both of the conditions, the equilibrium $\bar{E}$ of \eqref{system-delay} is asymptotically stable regardless of the value of the delay (independent asymptotic stability). However, we also need to note that if any of the conditions in Theorem \ref{positive-equi-asymptotically-stable-delay} is violated (particularly the inequalities in Point 2), the stability of the equilibrium will then depend on the delay value; and when the delay value varies, the equilibrium can lose stability, leading to oscillations 
\end{remark}

For example, if 
\begin{enumerate}
    \item If $m_3 < 0$: From equation \eqref{eq-z}, we would have that 
    \begin{equation}
        h(0) < 0, \quad \lim_{z \to + \infty} = +\infty,
    \end{equation}
    
    which means that \eqref{eq-z} has at least one positive real root, denoted by $\omega_0$.
    \item If $m_2 < 0$, we would have that 
    \begin{equation}
        \frac{-m_1 + \sqrt{m_1^2- 3m_2}}{2} = \frac{-3m_2}{m_1 + \sqrt{m_1^2 + 3m_2}} > 0
    \end{equation}
    
    which means that the equation \eqref{eq-z} has one positive real root $\omega_0$. 
\end{enumerate}

These two cases implies that the characteristic equation \eqref{characeq-delay} has a pair of purely imaginary roots $\pm i \omega_0$.

Next, we would focus on the bifurcation analysis, using the delay term $\tau$ as the bifurcation parameter, in light that the solutions of \eqref{characeq-delay} as function of this parameter. 

Let $\lambda(\tau) = \mu(\tau) + i \omega (\tau)$ be the eigenvalue of \eqref{system-realim} such that for some initial values of the bifurcation parameter $\tau_0$, we would have $\mu(\tau_0) = 0, \quad \omega(\tau_0) = \omega_0$. From the system \eqref{system-realim}, we would have:

\begin{equation}
    \tau_j = \frac{1}{\omega_0} \arccos \left( \frac{a_4 \omega_0^4 + (a_1a_3 - a_2a_4) \omega_0^2 - a_3a_5}{a_3^2 + a_4^2 \omega_0^2}\right) + \frac{2j\pi}{\omega_0}
\end{equation}

Moreover, we can verify the following transversal condition:

\begin{equation}
    \frac{d}{d\tau} \Re(\lambda(\tau)) \|_{\tau = \tau_0} = \frac{d}{d\tau} \mu(\tau) \|_{\tau = \tau_0} > 0
\end{equation}

holds. By continuity, the real part of $\lambda(\tau)$ becomes positive when $\tau > \tau_0$ and the steady state becomes unstable. Moreover, a Hopf bifurcation occurs when $\tau$ passes through the critical value $\tau_0$ (see \cite{hassard1981horfbifurcation}). 

To apply the Hopf bifurcation theorem stated in Marsden and McCracken \cite{marsden1976hopf}, we state the following theorem

\begin{theorem}
Suppose that $\omega_0$ is the largest positive simple root of \eqref{eq-omega}. Then, $i\omega (\tau_0) = i\omega_0$ is a simple root of \eqref{eq-omega}, and $\eta(\tau) + i\omega(\tau)$ is differentiable with respect to $\tau$ in a neighborhood of $\tau = \tau_0$
\end{theorem}

After previous reasoning, we admit that $i\omega_0$ is a simple root of \eqref{eq-omega}, which is an analytic equation; as a result, using the analytic version of the implicit function theorem mentioned in Chow and Hale (1982) \cite{chow1982bifurcation}, we have that $\eta(\tau) + i \omega(\tau)$ is well-defined and analytic in neighborhood of $\tau = \tau_0$

To establish the Hopf bifurcation at $\tau = \tau_0$, we need to show that 
\begin{equation}
    \frac{d\Re(\lambda(\tau))}{d\tau} \|_{\tau = \tau_0} > 0.
\end{equation}

In order to prove this inequality, we first start with a lemma and its respective proof.

\begin{lemma} \label{lemma-largest_positive_root}
Suppose that $z_1, z_2, z_3$ are the roots of $h(z) = z^3 + m_1 z^2 + m_2 z + m_3 = 0 \quad (m_2 < 0)$, and $z_3 \in \mathbb{R}^+$ is the largest positive simple root, then
\begin{equation}
    \frac{dh(z)}{dz} \|_{z = z_3} > 0.
\end{equation}
\end{lemma}

\begin{proof}
We will proceed to prove the lemma above with contradiction.

Assume that the largest positive simple root $z_3$ of the equation $h(z) = 0$ and 
\begin{equation}
    \frac{dh(z)}{dz} \|_{z = z_3} \leq 0.
\end{equation}

As a result, there exists a $\bar{z}  \in \mathbb{R}^+, \bar{z} > z_3$ such that $h(\bar{z}) < 0$. 

According to the Intermediate Value Theorem, for any $h \in \mathbb{R}, h \in [h(\bar{z}, + \infty)$, there always exists $z_4 \in [\bar{z}, + \infty)$. Taking $h = 0$, we would have 
\begin{equation}
    \begin{cases}
        h(z_4) = 0 \\
        z_4 > \bar{z} > z_3
    \end{cases} ,
\end{equation}

which means that $z_4$ is the highest positive simple root of $h(z) = 0$, not $z_3$. This contradicts our original assumption.

In conclusion, if $z_3$ is the largest positive simple root,

\begin{equation}
    \frac{dh(z)}{dz} \|_{z = z_3} > 0.
\end{equation}

The proof is complete.
\end{proof}

From the equation \eqref{characeq-delay}, derivating both sides with respect to $\tau$, we obtain

\begin{equation}
    \begin{aligned}
    \left( 3 \lambda^2 + 2 a_1 \lambda + a_2 \right) \frac{d\lambda}{d\tau} = & \left[ -\tau \exp(-\lambda \tau) (-a_3 - a_4 \lambda) + \exp(-\lambda \tau) (-a_4) \right] \frac{d\lambda}{d\tau} \\
    & - \lambda \exp(-\lambda \tau) (-a_3 - a_4 \lambda).
    \end{aligned}
\end{equation}

This gives us

\begin{equation}
    \begin{aligned}
    \left( \frac{d\lambda}{d\tau} \right)^{-1} & = \frac{3\lambda^2 + 2a_1 \lambda + a_2 + \tau \exp(-\lambda \tau) (-a_3 - a_4 \lambda) - \exp(-\lambda \tau) (-a_4)}{- \lambda \exp(-\lambda \tau) (-a_3 - a_4 \lambda)} \\
    & = \frac{3 \lambda^2 + 2a_1 \lambda + a_2}{-\lambda \exp(-\lambda \tau) (-a_3 - a_4 \lambda)} + \frac{a_4}{\lambda (a_3 + a_4 \lambda)} - \frac{\tau}{\lambda} \\
    & = \frac{2 \lambda^3 + a_1 \lambda^2 - a_5}{-\lambda^2(\lambda^3 + a_1 \lambda^2 + a_2 \lambda + a_5)} + \frac{-a_3}{\lambda^2 (a_3 + a_4 \lambda)} - \frac{\tau}{\lambda}
    \end{aligned}
\end{equation}

Thus, 

\begin{equation}
    \begin{aligned}
    \Sign \left\{ \frac{d(\Re(\lambda))}{d \tau} \right\} & = \Sign \left\{ \Re \left( \frac{d \lambda}{d \tau} \right)^{-1} \right\} \\
    & = \Sign \left\{ \Re \left[ \frac{2 \lambda^3 + a_1 \lambda^2 - a_5}{-\lambda^2 ( \lambda^3 + a_1 \lambda^2 + a_2 \lambda + a_5)}\right]_{\lambda = i \omega_0} + \Re \left[ \frac{-a_3}{\lambda^2(a_3 + a_4 \lambda)}\right]_{\lambda = i \omega_0} \right\} \\
    & = \Sign \left\{\Re \left[ \frac{-2\omega_0^3i - a_1 \omega_0^2 - a_5}{\omega_0^2(-\omega_0^3 i - a_1 \omega_0^2 + a_2 \omega_0 i + a_5)} \right] + \Re \left[ \frac{-a_3}{-\omega_0^2 (a_3 + a_4 \omega_0 i)}\right]\right\} \\
    & = \Sign \left\{ \frac{2 \omega_0^6 + (a_1^2 - 2a_2) \omega_0^4 - a_5^2}{\omega_0^2[(a_2\omega_0 - \omega_0^3)^2 + (a_5 - a_1 \omega_0^2)^2]} +  \frac{a_3^2}{\omega_0^2(a_4^2 \omega_0^2 + a_3^2)}\right\} \\
    & = \Sign \left\{ \frac{3\omega_0^4 + 2(a_1^2 - 2a_2) \omega_0^2 + (a_2^2 - 2a_1a_5 - a_4^2)}{(a_2 \omega_0 - \omega_0^3)^2 + (a_5 - a_1 \omega_0^2)^2} \right\}
    \end{aligned}
\end{equation}

Since 

\begin{equation}
    h(z) = z^3 + m_1z^2 + m_2z + m_3,
\end{equation}

we would have

\begin{equation}
    \frac{dh(z)}{dz} = 3z^2 + 2m_1z + m_2 = 3z^2 + 2(a_1^2 - 2a_2)z + (a_2^2 - 2a_1a_5 - a_4^2)
\end{equation}

As we have assumed that $\omega_0$ is the largest positive simple root of the equation \eqref{eq-omega}, from Lemma \ref{lemma-largest_positive_root}, we get

\begin{equation}
    \frac{dh(z)}{dz} \|_{z = \omega_0^2} > 0.
\end{equation}

Hence, 

\begin{equation}
    \frac{3\omega_0^4 + 2(a_1^2 - 2a_2) \omega_0^2 + (a_2^2 - 2a_1a_5 - a_4^2)}{(a_2 \omega_0 - \omega_0^3)^2 + (a_5 - a_1 \omega_0^2)^2} = \frac{\frac{dh(z)}{dz} \|_{z = \omega_0^2}}{(a_2 \omega_0 - \omega_0^3)^2 + (a_5 - a_1 \omega_0^2)^2} > 0
\end{equation}

or 

\begin{equation}
    \Sign \left\{ \frac{d(\Re(\lambda))}{d \tau} \right\} = \Sign \left\{ \frac{3\omega_0^4 + 2(a_1^2 - 2a_2) \omega_0^2 + (a_2^2 - 2a_1a_5 - a_4^2)}{(a_2 \omega_0 - \omega_0^3)^2 + (a_5 - a_1 \omega_0^2)^2} \right\} = 1
\end{equation}

i.e

\begin{equation}
    \frac{d(\Re(\lambda))}{d \tau} > 0
\end{equation}

The Hopf bifurcation analysis above can be summarized in the following theorem.

\begin{theorem} \label{bifurcation}

Suppose that 

\begin{equation}
    a_1 > 0, \quad a_3 + a_5 > 0, \quad a_1(a_2 + a_4) - (a_3 + a_5) > 0
\end{equation}

and 

\begin{equation}
    R_0 > 1
\end{equation}

If 

\begin{equation}
    m_3 < 0 \quad \lor \quad m_3 \geq 0, \quad m_2 < 0,
\end{equation}

the infected steady state $\bar{E}$ of the delay model \eqref{system-delay} is asymptotically stable when $\tau < \tau_0$ and unstable when $\tau > \tau_0$, where

\begin{equation}
    \tau_0 = \frac{1}{\omega_0} \arccos \left( \frac{a_4 \omega_0^4 + (a_1a_3 - a_2a_4) \omega_0^2 - a_3a_5}{a_3^2 + a_4^2 \omega_0^2}\right)
\end{equation}

When $\tau = \tau_0$, a Hopf bifurcation occurs; that is, a family of periodic solutions bifurcates from $\bar{E}$ as $\tau$ passes through the critical value $\tau_0$.
\end{theorem}

\section{Numerical simulation}

After providing all the analytical tools and qualitatively analysing the system for patterns on its dynamics, in this section ,we will perform some numerical analysis on the model to verify the previous results. 

\subsection{Simulation tools}

The numerical simulation is conducted on the programming language \texttt{Julia} through the package \texttt{DifferentialEquation.jl}, A Performant and Feature-Rich Ecosystem for Solving Differential Equations in Julia by Rackauckas and Nie (2017) \cite{rackauckas2017juliade}.

In order to avoid any stiffness in the ODE/DDE models, the algorithm for the Method of Steps in \texttt{Julia} is set to \texttt{Rosenbrock23}, which is the same as the classic ODE solver \texttt{ode23s} in \texttt{MATLAB}.

For the complete version of the \texttt{Julia} notebooks for simulation, please refer to the \texttt{Github} repository at \url{https://github.com/hoanganhngo610/DDE-HIV-NGO2020etal.}

\subsection{Simulation results}

\subsubsection{Simulation of the ODE model}

\begin{table}[H]
\centering
    \begin{tabular}{  l  l  l  } 
\hline
Parameters and Variables & & Values \\
\hline
\textit{Dependent variables} & & \\
$T$ & Uninfected CD4\textsuperscript{$+$} T-cell population size & $250$ mm\textsuperscript{$-3$}\\
$I$ & Infected CD4\textsuperscript{$+$} T-cell density & $50$ mm\textsuperscript{$-3$} \\
$V$ & Initial density of HIV RNA & $160$ mm\textsuperscript{$-3$} \\
\textit{Parameters and Constants} & & \\
$s$ & Source term for uninfected CD4\textsuperscript{$+$} T-cells & $5$ day\textsuperscript{$-1$} mm\textsuperscript{$-3$} \\
$d$ & Natural death rate of CD4\textsuperscript{$+$} T-cells & 0.01 day\textsuperscript{-1} \\
$a$ & Growth rate of CD4\textsuperscript{$+$} T-cell population & $0.8$ day\textsuperscript{$-1$} \\
$T_{\max}$ & Maximal population level of CD4\textsuperscript{$+$} T-cells & $1500$ mm\textsuperscript{$-3$}\\ 
$\beta$ & Rate CD4\textsuperscript{$+$} T-cells became infected with virus & $2.4 \times 10^{-4}$ mm\textsuperscript{$-3$} \\
$\alpha$ & Saturated mass-action term & $0.001$ \\
$\rho$ & Rate of cure & $0.01$ day\textsuperscript{$-1$} \\
$\delta$ & Blanket death rate of infected CD4\textsuperscript{$+$} T-cells & $0.3$ day\textsuperscript{$-1$}\\
$q$ & Reproduction rate of the infected CD4\textsuperscript{$+$} T-cells & $500$ mm\textsuperscript{$-3$} day\textsuperscript{$-1$} \\
$c$ & Death rate of free virus & $8$ day\textsuperscript{$-1$} \\ 
\hline
\end{tabular}
\caption{\label{tab:table_1} Preliminary values of variables and parameters for viral spread.}
\end{table}

\begin{table}[H]
\centering
    \begin{tabular}{  l  l  l  l  l } 
\hline
Parameters & Original scenario & Scenario \#2 & Scenario \#3 & Scenario \#4 \\
$s$ & $5$ & $-$ & $-$ & $-$\\
$d$ & $0.01$ & $-$ & $-$ & $-$ \\
$a$ & $0.8$ & $8$ & $-$ & $-$ \\
$T_{\max}$ & $1500$ & $-$ & $-$ & $-$ \\
$\beta$ & $2.4 \times 10^{-4}$ & $-$ & $0.0024$ & $0.0024$ \\
$\alpha$ & $0.001$ & $0.0001$ & $0.000001$ & $0.000001$ \\
$\rho$ & $0.01$ & $0.01$ & $-$ & $-$ \\
$\delta$ & $0.3$ & $5$ & $-$ & $-$ \\
$q$ & $500$ & $-$ & $2.5$ & $2.5$ \\
$c$ & $8$ & $1.3$ & $3$ & $1.3$ \\

\hline

\end{tabular}
\caption{\label{tab:table_2} Values of parameters for viral spread in different scenarios.}
\end{table}

Within the range of parameters that are proven to be realistic in medical research, we investigate the behavior of the model within 4 different scenarios.

\begin{itemize}
    \item \textbf{The original scenario:} In this scenario, the condition $1$, $2$ and $3$ in Theorem \ref{globally-asymptotically-stable} are satisfied. This means that, the positive equilibrium of the system \eqref{system} is globally asymptotically stable.
    
    \begin{figure}[H]
    \centering
    \caption{The ODE model is locally asymptotically stable with parameters in the \textbf{original scenario}}
    \includegraphics[scale=0.05]{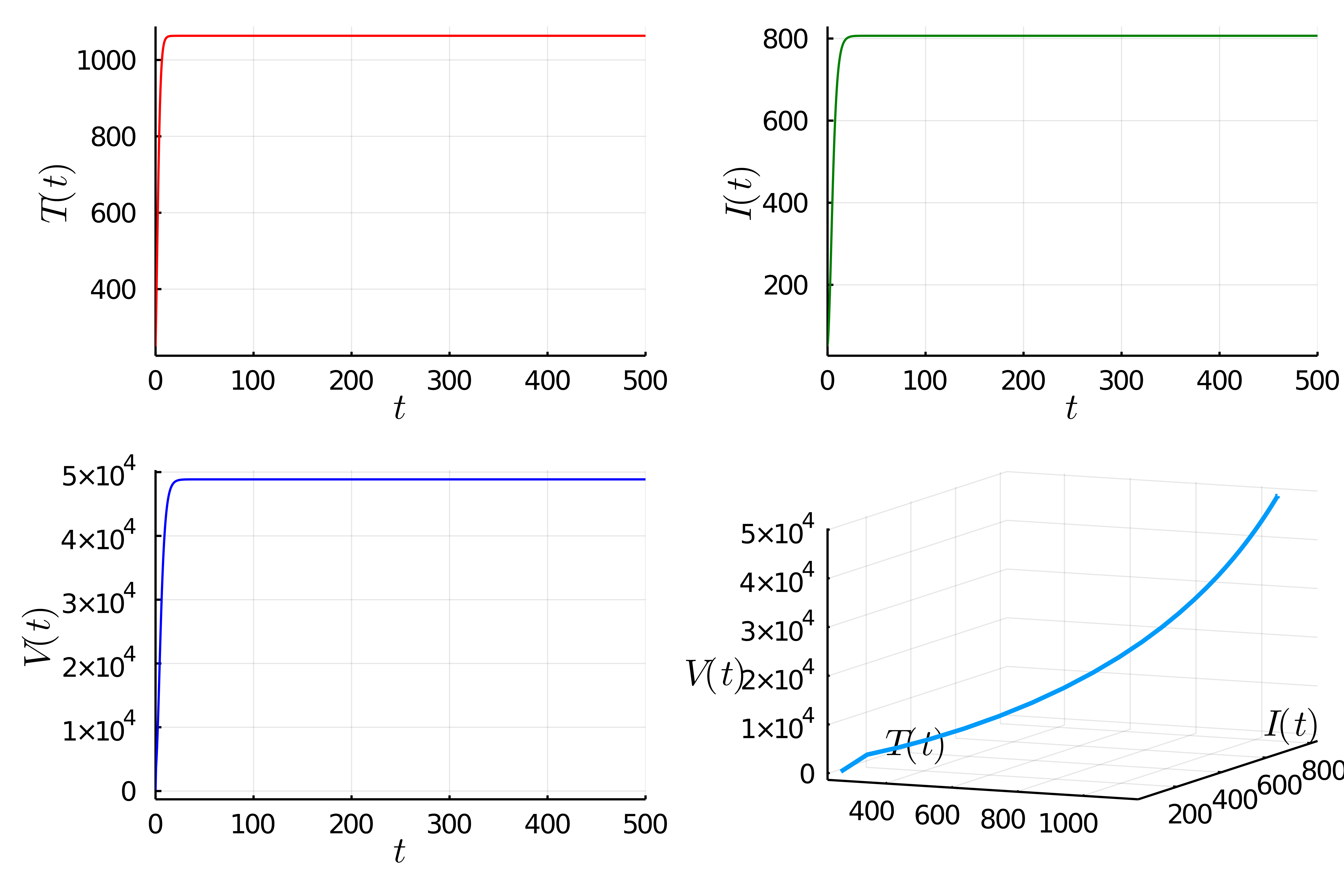}
    \label{fig:fig1}
    \end{figure}
    
    \item \textbf{Scenario \#2:} In this scenario, the conditions $1$, $2$ and $4$ in Theorem \ref{globally-asymptotically-stable} are satisfied. This means that, the positive equilibrium of the system \eqref{system} is also globally asymptotically stable.
    
    \begin{figure}[H]
    \centering
    \caption{The ODE model is locally asymptotically stable with parameters in \textbf{Scenario \#2}}
    \includegraphics[scale=0.05]{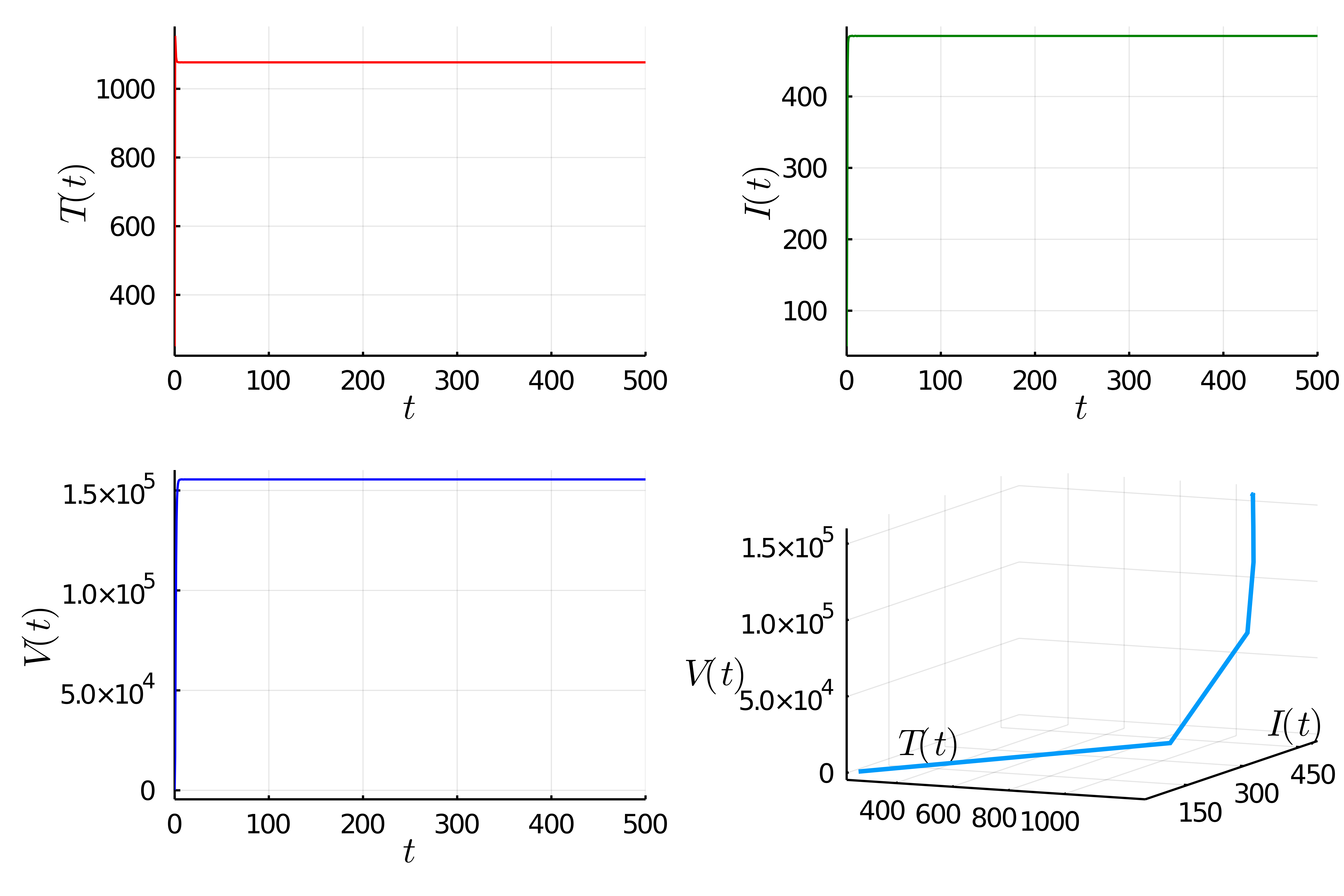}
    \label{fig:fig2}
    \end{figure}
    
    \item \textbf{Scenario \#3:} In this scenario, the conditions $1$ and $2$ of Theorem \ref{asymptotically-stable} is satisfied. This means that, the positive equilibrium of the system \eqref{system} is locally asymptotically stable.
    
    \begin{figure}[H]
    \centering
    \caption{The ODE model is locally asymptotically stable with parameters in \textbf{Scenario \#3}}
    \includegraphics[scale=0.05]{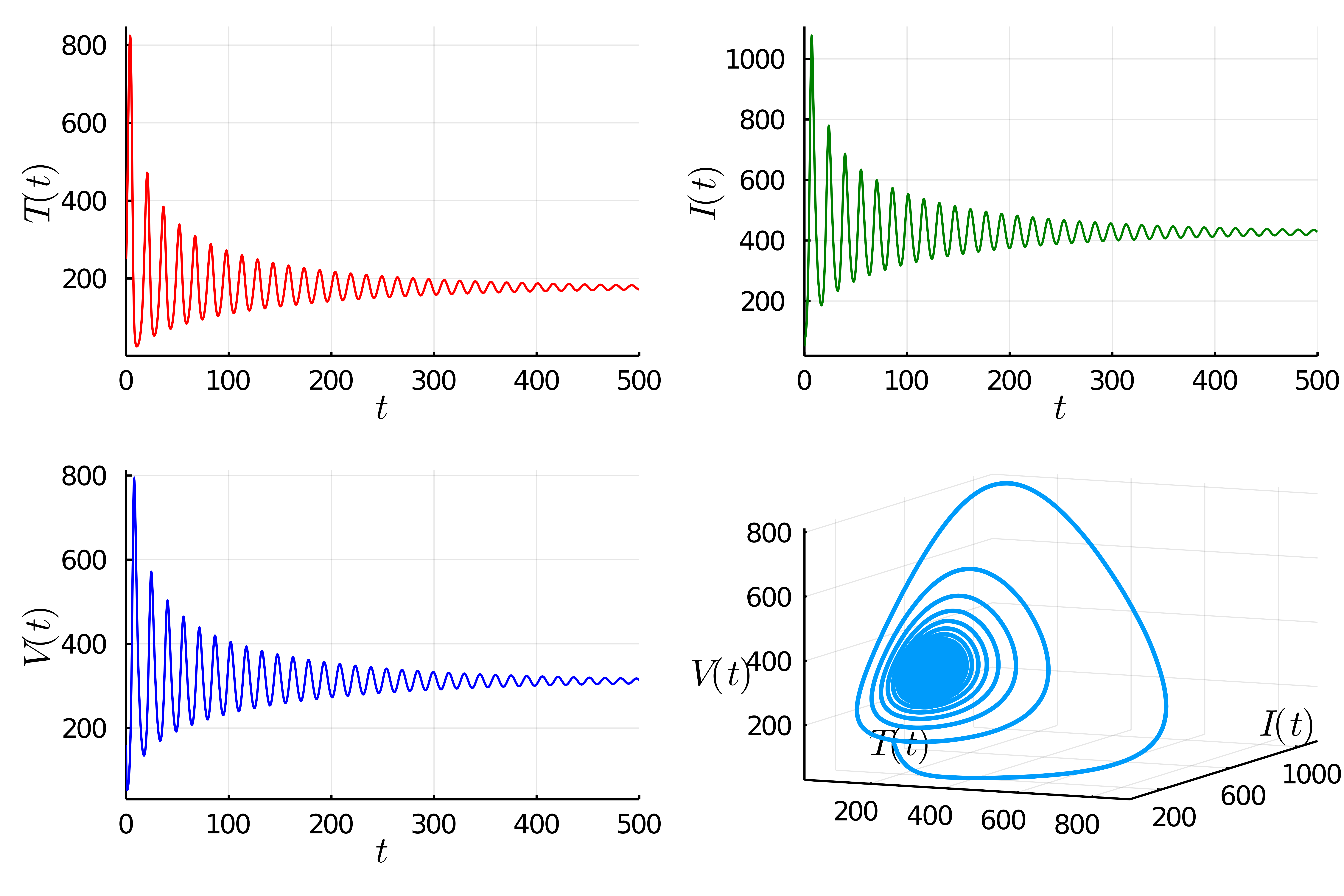}
    \label{fig:fig3}
    \end{figure}
    
    \item \textbf{Scenario \#4:} In this scenario, the conditions $1$ and $2$ of Theorem \ref{orbitally-asymptotically-stable} is satisfied. This means that, the positive equilibrium of the system \eqref{system} is orbitally asymptotically stable.
    
    \begin{figure}[H]
    \centering
    \caption{The ODE model is orbitally asymptotically stable with parameters in \textbf{Scenario \#4}}
    \includegraphics[scale=0.05]{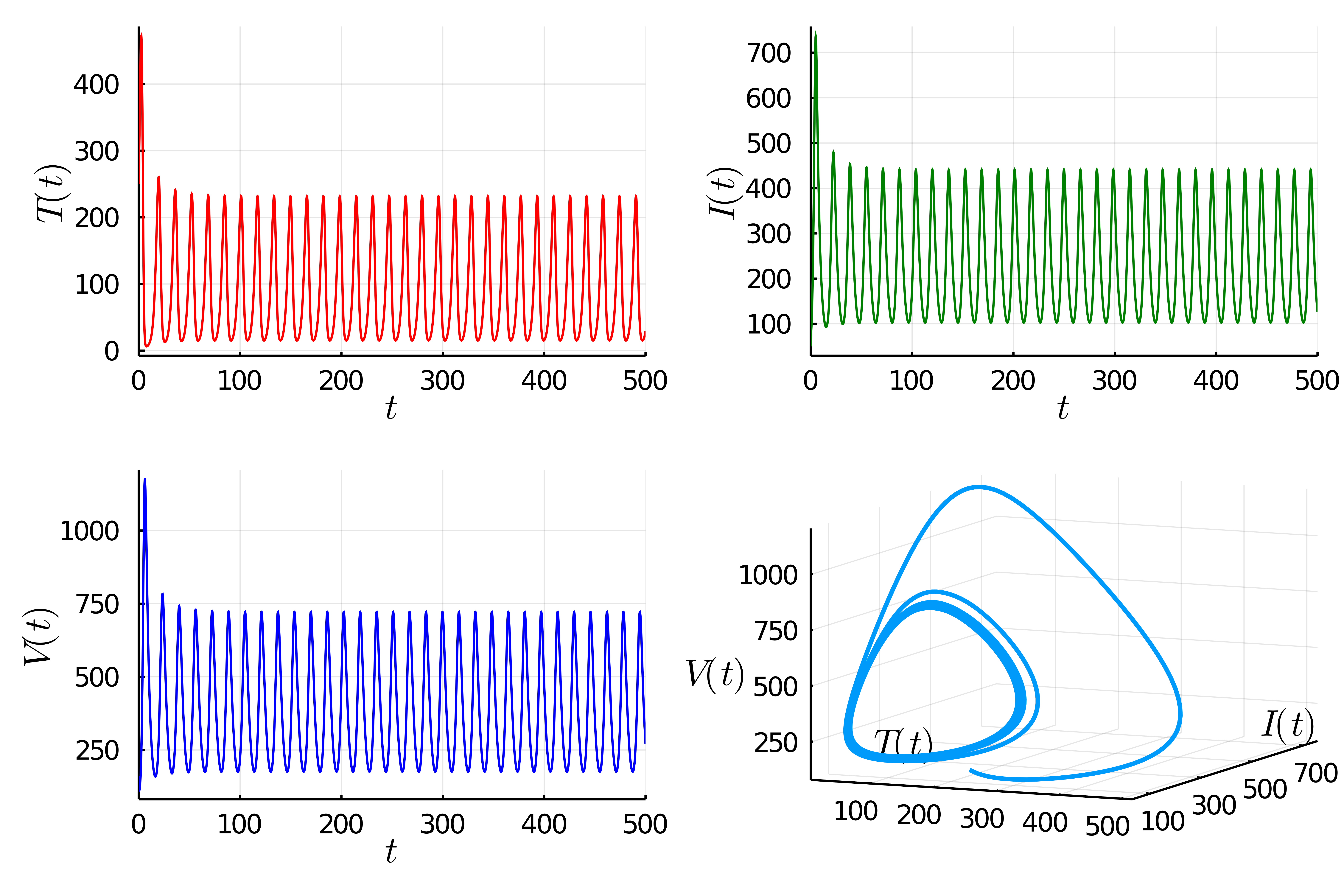}
    \label{fig:fig4}
    \end{figure}
    
\end{itemize}

\subsubsection{Simulation of the DDE model}

\begin{table}[H]
\centering
    \begin{tabular}{  l  l  l  l  l l } 
\hline
Parameters (DDE) & Original scenario (ODE) & Scenario \#1 & Scenario \#2 & Scenario \#3 & Scenario \#4 \\
$s$ & $5$ & $-$ & $-$ & $-$ & $-$\\
$d$ & $0.01$ & $-$ & $-$ & $-$ & $-$ \\
$a$ & $0.8$ & $-$ & $5$ & $5$ & $5$ \\
$T_{\max}$ & $1500$ & $-$ & $-$ & $-$ & $-$ \\
$\beta$ & $2.4 \times 10^{-4}$ & $-$ & $-$ & $-$ & $-$\\
$\alpha$ & $0.001$ & $0.000001$ & $0.000001$ & $0.000001$ & $0.000001$ \\
$\rho$ & $0.01$ & $-$ & $-$ & $-$ & $0.3$ \\
$\delta$ & $0.3$ & $-$ & $-$ & $-$ & $-$ \\
$q$ & $500$ & $-$ & $-$ & $-$ & $-$ \\
$c$ & $8$ & $-$ & $-$ & $-$ & $-$ \\
$\tau$ & N/A & $0.4$ & $10$ & $5$ & $5$ \\
\hline

\end{tabular}
\caption{\label{tab:table_3} Values of parameters for viral spread in different scenarios.}
\end{table}

First of all, instead of keeping $\alpha = 0.001$, we modify this parameter into $\alpha = 0.000001$ so that we can observe different behaviors while modifying other parameters.

\begin{itemize}
    \item \textbf{Scenario \#1 (DDE):} When $\tau = 0.4$ and all other variables are kept the same as the original scenario in the ODE setting (apart from $\alpha$), $T(t)$, $I(t)$ and $V(t)$ all converges to their positive equilibrium. We say that, in this setting, the positive equilibrium $\bar{E}$ is globally asymptotically stable.
    
    \begin{figure}[H]
    \centering
    \caption{The DDE model is globally asymptotically stable with parameters in \textbf{Scenario \#1}, with the delay term $\tau = 0.4$.}
    \includegraphics[scale=0.05]{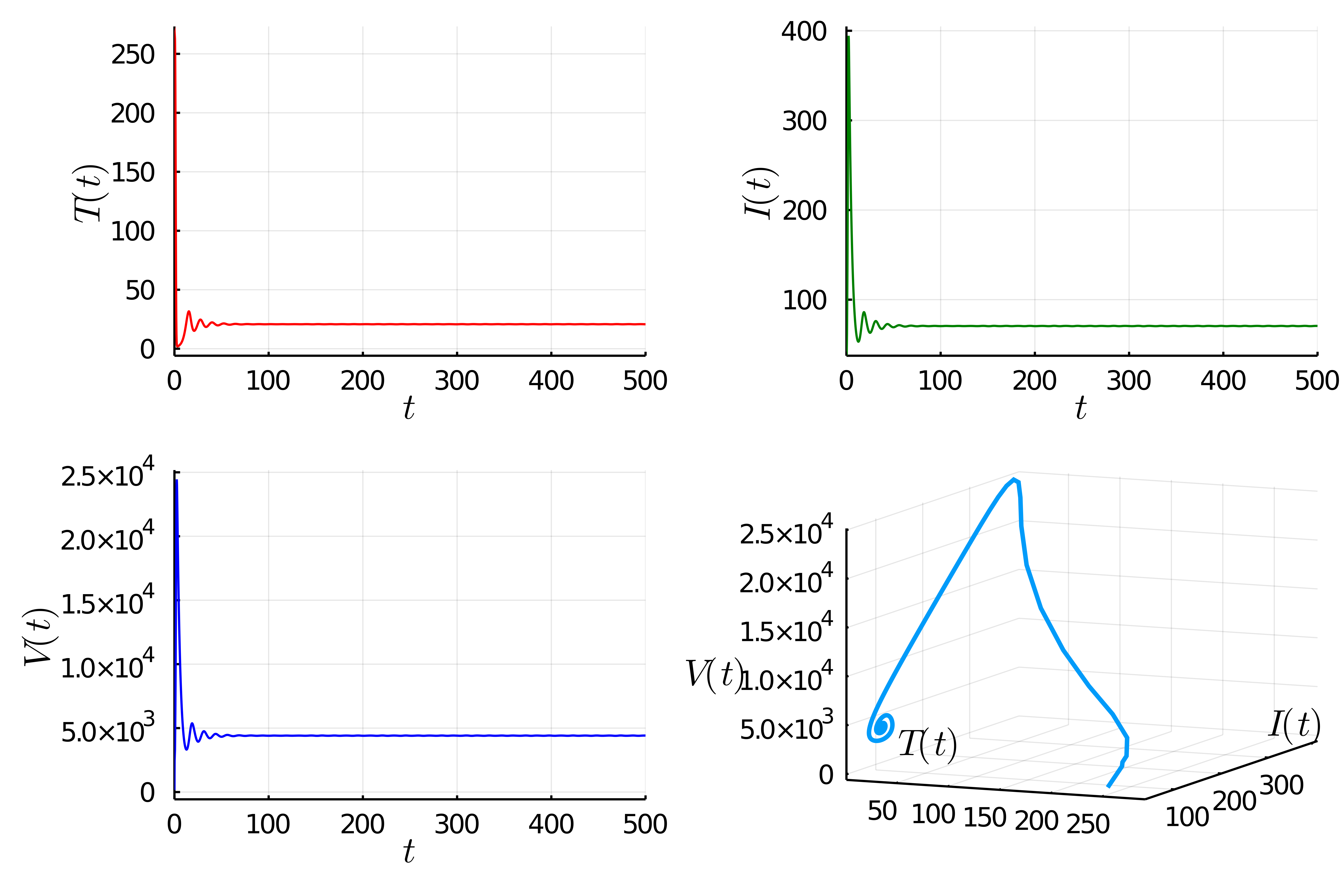}
    \label{fig:fig1-DDE}
    \end{figure}
    
    \item \textbf{Scenario \#2 (DDE):} When we modify $a = 5$ and $\tau = 10$, the parameters, in this setting, satisfy the conditions of Theorem \ref{bifurcation}. This means that, the positive equilibrium $\bar{E}$ is orbitally asymptotically stable, or in other words, there exists a positive periodic solution for all the components of the system.
    
    \begin{figure}[H]
    \centering
    \caption{The DDE model is globally asymptotically stable with parameters in \textbf{Scenario \#2}, with $a = 5$ and the delay term $\tau = 10$.}
    \includegraphics[scale=0.05]{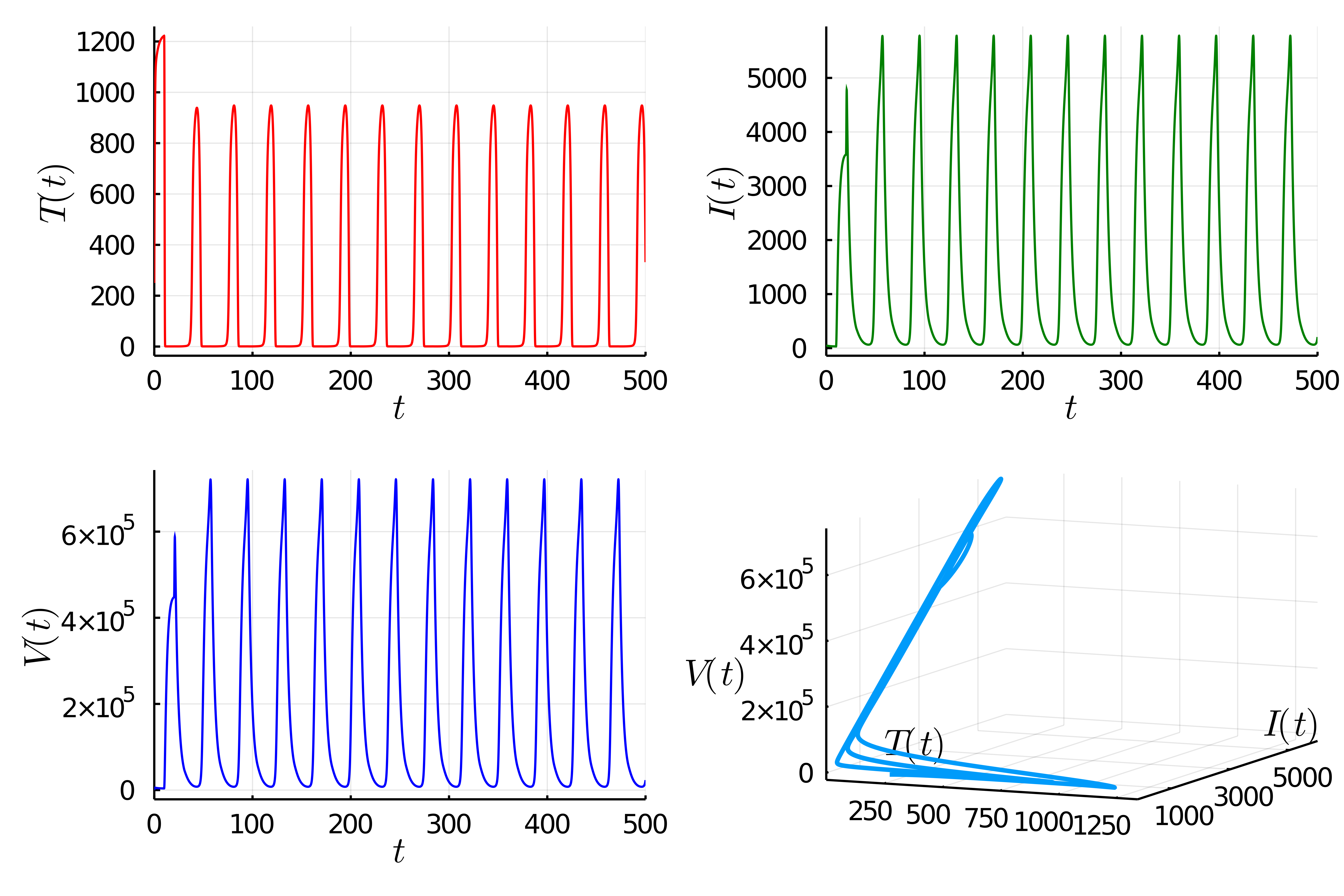}
    \label{fig:fig2-DDE}
    \end{figure}
    
     \item From the last two scenarios, \textbf{Scenario \#3 (DDE)} and \textbf{Scenario \#4 (DDE)}, we can draw a conclusion that the solution of the system would return to stability when the cure rate $\rho$ is increased. For example, if we select $\rho = 0.3$ instead of $\rho = 0.01$ with all other parameters kept identical, the system would admit a lower global asymptotic stability. We can conclude that $\rho$ is an important parameter in the sense that increasing it helps us control the disease.
    
    \begin{figure}[H]
    \centering
    \caption{The DDE model is globally asymptotically stable with parameters in \textbf{Scenario \#3} and \textbf{Scenario \#4}, with $\rho = 0.01$ and $\rho = 0.3$, respectively. These graphs show that the cure rate is an important parameter in controlling the disease.}
    \includegraphics[scale=0.05]{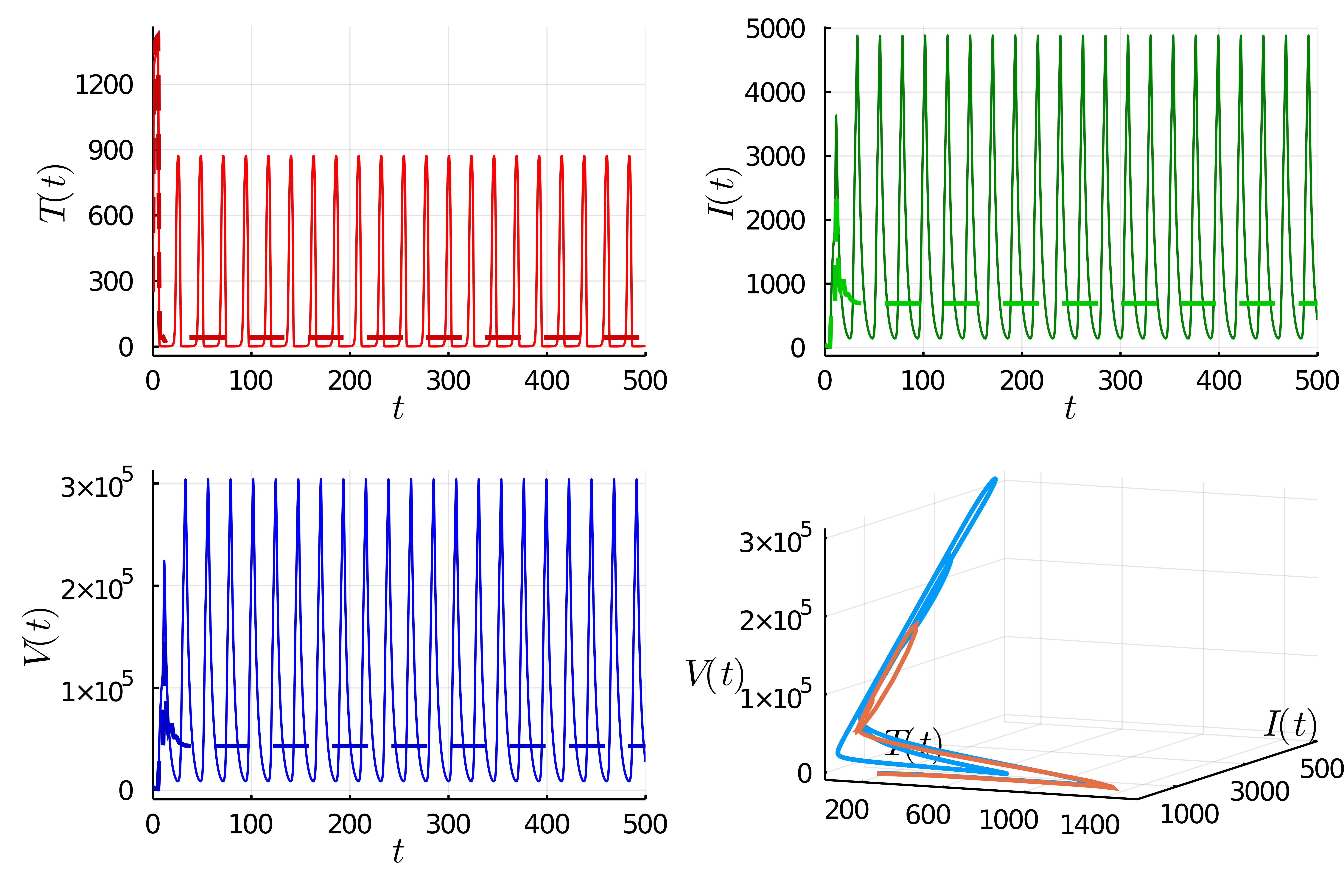}
    \label{fig:fig3-DDE}
    \end{figure}
    
\end{itemize}

\appendix
\section{List of macros for formatting text, figures and tables}

\begin{theorem}[Gronwall, 1919] Let $I$ denote an interval of the real line of the form $[a, \inf)$ or $[a,b]$lr $[a,b)$ with $a < b$ Let $\beta$ and $u$ be real-valued continuous functions defined on $I$. If $u$ is a differentiable in the interior $I^0$ of $I$ (the interval $I$ without the end points $a$ and possibly $b$) and satisfies the differential inequality
\begin{equation}
    u'(t) \leq \beta(t) u(t),  t \in I^0
\end{equation}
then $u$ is bounded by the solution of the corresponding differential equation $\nu'(t) = \beta(t) \nu(t)$:
\begin{equation}
    u(t) \leq u(a) \exp \left( \int_a^t \beta(s) ds \right)
\end{equation}
\end{theorem}

\begin{theorem}[Lyapunov's stability]
Let a function $V(\mathbf{X})$ be continuously differentiable in a neighbourhood $U$ of the origin. The function $V(\mathbf{X})$ is called the Lyapunov function for an autonomous system

\begin{equation}
    \mathbf{X}' = f(\mathbf{X})
\end{equation}

if the following conditions are met:
\begin{enumerate}
    \item $V(\mathbf{X}) > 0$ for all $\mathbf{X} \in U \setminus \{ 0 \}$;
    \item V(0) = 0;
    \item $\frac{dV}{dt} \leq 0$ for all $\mathbf{X} \in U$.
\end{enumerate}
Then, if in a neighborhood $U$ of the zero solution $\mathbf{X} = 0$ of an autonomous system there is a Lyapunov function $V(\mathbf{X})$ with a negative definite derivative $\frac{dV}{dt}$ for all $\mathbf{X} \in U \setminus \{ 0 \}$, then the equilibrium point $\mathbf{X} = 0$ of the system is asymptotically stable. 
\end{theorem}

\begin{theorem}[Perron - Frobenius theorem]  \cite{gantmacher1959theorymatrices} Let $A$ be a irreducible Metzler matrix (A Metzler matrix is a matrix whose all of its off-diagonal elements are non-negative). Then, $\lambda_M$, the eigenvalue of $A$ of largest real part is real, and the elements of its associated eigenvector $v_M$ are positive. Moreover, any eigenvector of $A$ with non-negative elements belongs the the span of $v_M$.
\end{theorem}

\begin{theorem}[Implicit Function Theorem (Chow and Hale, 1982)]
Suppose that 
\begin{itemize}
    \item $X, Y, Z$ are Banach spaces,
    \item $F: U \times V \to Z$ is continuously differentiable,
    \item $F(x_0, y_0) = 0$ and $D_xF(x_0, y_0)$ has a bounded inverse.
\end{itemize}

Then, there exists a neighborhood $U_1 \times V_1 \in U \times V$ of $(x_0, y_0)$ and a function $f:V_1 \to U_1, f(y_0) = x_0$ such that 

\begin{equation}
    F(x,y) = 0 \text{ for } (x,y) \in U_1 \times V_1 \text{ iff } x = f(y)
\end{equation}
If $F \in C^k(U \times V, Z), k \geq 1$ or analytic in a neighborhood $(x_0, y_0)$, then $f \in C^k(V_1, X)$ or is analytic in a neighborhood of $y_0$.

\end{theorem}

\begin{theorem}[Poincaré - Bendixson theorem] \cite{wiki:poincare-bendixson}

Given a differentiable real dynamical system defined on an open subset of the plane, every non-empty compact $\omega$-limit set of an orbit, which contains only finitely many fixed points, is either

\begin{itemize}
    \item a fixed point,
    \item a periodic orbit, or
    \item a connected set composed of a finite number of fixed points together with homoclinic and heteroclinic orbits connecting these.
\end{itemize}

Moreover, there is at most one orbit connecting different fixed points in the same direction. However, there could be countably many homoclinic orbits connecting one fixed point.

\end{theorem}

Next, we will give the definition of an additive compound matrix and consider the particular case when it's a square matrix \cite{song2007hivmodel}. A survey of properties of additive compound matrices, along with their connections to differential equations have been investigated in \cite{muldowney1990compoundmat, li1995globalstabilityseir} 

We will start with  the definition of the $k$-th exterior power (or multiplicative compound) of an $n \times m$ matrix.

\begin{definition} [Multiplicative compound of a matrix]
Let $A$ be an $n \times m$ matrix of real or complex numbers. Let $a_{i_1, i_2, ..., i_k, j_1, j_2, ..., j_k}$ be the minor of $A$ determined by the rows $(i_1, ..., i_k)$ and the columns $(j_1, ..., j_k)$, $1 \leq i_1 < i_2 < ... < i_k \leq n, 1 \leq j_1 < j_2 < ... < j_k \leq m$. The $k$-th multiplicative compound matrix $A^{(k)}$ of $A$ is the $\binom{n}{k} \times \binom{m}{k}$ matrix whose entries, written in lexicographic order, are $a_{i_1, ..., i_k, j_1, ..., j_k}$.
\end{definition}

In particular, when $A$ is an $n \times k$ matrix with columns $a_1, a_2, ..., a_k$, $A^{(k)}$ is the exterior product $a_1 \lor a_2 \lor ... \lor a_k$.

In the case $m = n$, the additive compound matrices are defined as follows.

\begin{definition}
Let $A$ be an $n \times n$ matrix. The $k$-th additive compound $A^{[k]}$ of $A$ is the $\binom{n}{k} \times \binom{n}{k}$ matrix given by

\begin{equation} \label{compound-matrix}
    A{[k]} = D(I + hA) \|_{h = 0}
\end{equation}

If $B = A^{[k]}$, the following formula for $b_{i,j}$ can be deduced from the equation \eqref{compound-matrix}, For any integer $i = 1, ..., \binom{n}{k}$, let $(i) = (i_1, i_2, ..., i_k)$ be the $i$-th member in the lexicographic ordering of all $k$-tuples of integers such that $1 \leq i_1 < i_2 < ... < i_k \leq n$. Then,

\begin{equation}
    b_{i,j} = 
    \begin{cases}
        \begin{aligned}
            a_{i_1, i_1} + ... + a_{i_k, i_k} \quad & \text{ if } (i) = (j) \\
            (-1)^{r+s} a_{i_s, j_r} \quad & \text{ if exactly one entry $i_s$ in $(i)$ does not occur in $(j)$ and $j_r$ does not occur in $(i)$, } \\
            0 \quad & \text{ if $(i)$ differs from $(j)$ in two or more entries.}
        \end{aligned}
    \end{cases}
\end{equation}
\end{definition}

In the extreme cases when $k = 1$ and $k = n$, we would have that $A^{[1]} = A$ and $A^{[n]} = \text{tr}(A)$. For $n = 3$ ,which is the case that we are considering in this paper, we would have the matrices $A^{[k]}, k = 0, 1, 2$ as follows:

\begin{equation}
    A^{[1]} = A, \quad 
    A^{[2]} = \left(
    \begin{matrix} 
    a_{11} + a_{22} & a_{23} & -a_{13} \\
    a_{32} & a_{11} + a_{33} & a_{12} \\
    -a_{31} & a_{21} & a_{22} + a_{33}, 
    \end{matrix} \right), \quad
    A^{[3]} = a_{11} + a_{22} + a_{33}
\end{equation}

\section*{Conflicts of Interest}
The authors declare that there are no conflicts of interest regarding the publication of this paper.

\section*{Acknowledgement}

The authors would like to thank Nguyen Tran Hai Yen, undergraduate student at the 
Faculty of Biology and Biotechnology, Ho Chi Minh University of Science, VNU - HCM, Class of 2022 for providing valuable biological insights and ideas to support this research.




\nocite{*}

\printbibliography
\end{document}